\documentclass[11pt,letterpaper]{article}
\usepackage{times,mathptmx}
\DeclareMathAlphabet{\mathcal}{OMS}{cmsy}{m}{n}
\usepackage{fullpage}
\usepackage{amsmath,amsfonts,amssymb,amsthm,amscd}
\usepackage{tabularx,arydshln}
\usepackage{graphicx,epstopdf}
\usepackage{hyperref,cite}
\newcommand{\vs}{\vspace{1.5mm}}
\newtheorem{theorem}{Theorem}[section]

\newtheorem{lemma}[theorem]{Lemma}

\newtheorem{definition}[theorem]{Definition}

\newcommand{\G}{\mathbb{G}}
\newcommand{\Z}{\mathbb{Z}}

\newcommand{\mcA}{\mathcal{A}}
\newcommand{\mcB}{\mathcal{B}}
\newcommand{\mcC}{\mathcal{C}}
\newcommand{\Adv}{\textbf{Adv}}

\newcommand{\db}{\displaybreak[0]}

\title{Anonymous HIBE with Short Ciphertexts:\\ Full Security in Prime
    Order Groups\footnote{
        This work was partly supported by
        the MSIP (Ministry of Science, ICT \& Future Planning), Korea, under
        the C-ITRC (Convergence Information Technology Research Center)
        support program (NIPA-2013-H0301-13-3007) supervised by the NIPA
        (National IT Industry Promotion Agency)
        and
        the IT R\&D program of MOTIE/KEIT [KI002113, Development of Security
        Technology for Car-Healthcare].}
    }

\author{
    Kwangsu Lee\footnote{Korea University, Korea and Columbia University, USA.
        Email: \texttt{guspin@korea.ac.kr}.}
    \and
    Jong~Hwan Park\footnote{Korea University, Korea and Sangmyung University,
        Korea. Email: \texttt{decartian@korea.ac.kr}.}
    \and
    Dong~Hoon Lee\footnote{Korea University, Korea.
        Email: \texttt{donghlee@korea.ac.kr}.}
}

\date{}

\begin{document}

\maketitle

\begin{abstract}
Anonymous Hierarchical Identity-Based Encryption (HIBE) is an extension of
Identity-Based Encryption (IBE), and it provides not only a message hiding
property but also an identity hiding property. Anonymous HIBE schemes can be
applicable to anonymous communication systems and public key encryption
systems with keyword searching. However, previous anonymous HIBE schemes have
some disadvantages that the security was proven in the weaker model, the size
of ciphertexts is not short, or the construction was based on composite order
bilinear groups. In this paper, we propose the first efficient anonymous HIBE
scheme with short ciphertexts in prime order (asymmetric) bilinear groups,
and prove its security in the full model with an efficient reduction. To
achieve this, we use the dual system encryption methodology of Waters. We
also present the benchmark results of our scheme by measuring the performance
of our implementation.
\end{abstract}

\vs \noindent {\bf Keywords:} Identity-based encryption, Hierarchical
identity-based encryption, Anonymity, Full model security, Bilinear maps.

\newpage

\section{Introduction}

Hierarchical Identity-Based Encryption (HIBE) is an extension of
Identity-Based Encryption (IBE) that uses an identity as a public key. In
HIBE, a user's identity is represented as a hierarchical tree structure and
an upper level user can delegate the private key generation capability to a
lower level user. Horwitz and Lynn introduced the concept of HIBE to reduce
the burden of the private-key generator of IBE \cite{HorwitzL02}. After the
introduction of HIBE, it was shown that HIBE can have various applications
like identity-based signature \cite{GentryS02}, public-key broadcast
encryption \cite{DodisF02}, forward-secure public key encryption
\cite{CanettiHK03}, and chosen-ciphertext secure HIBE \cite{CanettiHK04}.

Recently, as a result of the increasing concern with users' privacy, the need
for cryptographic systems that protect users' privacy also increases.
Anonymous HIBE can provide users' privacy by supporting not only the
\textit{message hiding} property but also the \textit{identity hiding}
property that hides identity information in ciphertexts. Abdalla et al.
formalized the concept of anonymous HIBE \cite{AbdallaBC+05}. After that,
Boyen and Waters proposed the first secure anonymous HIBE scheme without
random oracles \cite{BoyenW06}. The main applications of anonymous HIBE are
anonymous communication systems that provide anonymity between a received
message and a true sender and public key encryption systems with keyword
searching that enable keyword searches on encrypted data \cite{BonehCOP04}.

The security model of anonymous HIBE is defined as a game between a
challenger and an adversary. In this game, the adversary adaptively requests
private keys in the private key query step and selects two hierarchical
identities $ID_0, ID_1$ and two messages $M_0, M_1$ in the challenge step.
Next, the adversary is given a challenge ciphertext of $ID_{\gamma},
M_{\gamma}$ where ${\gamma}$ is a random bit chosen by the challenger. The
adversary wins the game if he can correctly guess $\gamma$. The security
model is divided as a selective model where the adversary should commit the
target hierarchical identities in the initial step and a full model where the
adversary can select the target hierarchical identities in the challenge
step. Generally a selectively secure HIBE scheme is converted to a fully
secure HIBE scheme, but the reduction is inefficient \cite{BonehB04e}. The
efficiency of the reduction is important not only for theoretical reasons but
also for practical reasons.

Let $\Adv_{\mcA}$ be the advantage of an adversary $\mcA$ that breaks a
scheme and $\Adv_{\mcB}$ be the advantage of an algorithm $\mcB$ that breaks
an assumption using the adversary $\mcA$. Suppose that $\Adv_{\mcA} \leq L
\cdot \Adv_{\mcB}$ where $L$ is a reduction loss. Let $\lambda, k$ be the
security level of the scheme and the assumption, respectively. If the
assumption provides the $k$-bit security, then it guarantees that
$\Adv_{\mcB} \leq 1/2^k$ for any PPT algorithm $\mcB$. Then we can derive
$\Adv_{\mcA} \leq L \cdot 1/2^k$ from two inequalities $\Adv_{\mcA} \leq L
\cdot \Adv_{\mcB}$ and $\Adv_{\mcB} \leq 1/2^k$. To construct the scheme that
provides the $\lambda$-bit security, it should be guaranteed that
$\Adv_{\mcA} \leq 1/2^{\lambda}$ for any PPT adversary $\mcA$. It is easy to
achieve this by setting $L \cdot 1/2^k \leq 1/2^{\lambda}$ since $\Adv_{\mcA}
\leq L \cdot 1/2^k$. Thus we can derive a relation $k \geq \lambda +
\log_2(L)$. This relation says that the bit size $k$ of a group order for the
assumption should be larger than $\lambda + \log_2(L)$ to construct the
scheme with the $\lambda$-bit security. For example, if there is a
selectively secure scheme with a hierarchical depth $l=10$, then we should
select $k = 880$ since $\lambda = 80$ and $L=2^{\lambda l}$. Therefore, an
ideal anonymous HIBE scheme should be fully secure with a reduction loss less
than $c \cdot q$ for a polynomial value $q$ and a constant $c$.

To construct a fully secure HIBE scheme with an efficient reduction, the new
proof methodology named the dual system encryption method was proposed by
Waters \cite{Waters09}. In the dual system encryption method, ciphertexts and
private keys can be a normal type or a semi-functional type, and the
semi-functional types of ciphertexts and private keys are only used in
security proofs. Additionally, the normal type and the semi-functional type
are indistinguishable, and the semi-functional ciphertexts are not decrypted
by using the semi-functional private keys. The proof of the dual system
encryption method consists of hybrid games that change a normal ciphertext
and normal private keys to a semi-functional ciphertext and semi-functional
private keys. Using this methodology, Waters proposed a fully secure HIBE
scheme with linear-size ciphertexts and a fully secure HIBE scheme with
constant-size ciphertexts \cite{Waters09,LewkoW10}. The dual system
encryption method can be used to prove the security of fully secure
attribute-based encryption \cite{LewkoOSTW10}, fully secure predicate
encryption \cite{OkamotoT10}, and leakage-resilient cryptography
\cite{LewkoRW11}.

The first secure anonymous HIBE scheme was proposed by Boyen and Waters
\cite{BoyenW06}, and it was proven to be selectively secure without random
oracles. After the first construction of anonymous HIBE, several anonymous
HIBE schemes were presented, but they were only proved to be secure in the
selective model \cite{SeoKOS09,Ducas10,LeeL10}. Recently, De Caro et al.
proposed a fully secure anonymous HIBE scheme with short ciphertexts by using
the dual system encryption method \cite{CaroIP10}. However, their scheme is
inefficient since the scheme is based on composite order groups where the
group order is a product of four prime numbers. One may use the conversion
method of Freeman \cite{Freeman10} to construct a scheme in prime order
groups from a scheme in composite order groups, but this method can not be
applied to the dual system encryption method of Lewko and Waters
\cite{LewkoW10} since it does not provide the parameter hiding property in
composite order groups\footnote{Lewko and Waters used the parameter hiding
property of composite order groups to prove the full security of their HIBE
scheme using the dual system encryption technique \cite{LewkoW10}. The
parameter hiding property of composite order $N=pqr$ is stated that an
exponent $\Z_N$ has one-to-one correspondence with $(\Z_p, \Z_q, \Z_r)$
because of Chinese Remainder Theorem (CRT) and $\Z_q$ and $\Z_r$ values are
information theoretically hidden to an adversary even if $\Z_p$ value is
revealed to the adversary.}. Lewko recently devised another conversion method
for the dual system encryption method and constructed a (non-anonymous)
unbounded HIBE scheme with linear-size ciphertexts in prime order groups
\cite{Lewko12}. However, this method is not known to be applicable for the
construction of an anonymous HIBE scheme with constant-size ciphertexts since
it uses dual pairing vector spaces (DPVS)\footnote{The dimensions of DPVS is
generally proportional to the size of an identity vector in the scheme that
uses DPVS \cite{OkamotoT09,Lewko12,OkamotoT12}. Thus an HIBE scheme based on
DPVS that supports $l$-depth has linear-size of ciphertexts since it requires
at least $l$-dimensions in DPVS. To reduce the dimensions of DPVS, one may
try to use the technique of Okamoto and Takashima \cite{OkamotoT11}, but it
only applied to non-anonymous schemes since it should reveal the identity of
ciphertexts.}.

Anonymous HIBE can also be constructed from Predicate Encryption (PE) with
the delegation capability. Shi and Waters constructed an anonymous HIBE
scheme with linear-size ciphertexts from a delegatable Hidden Vector
Encryption (dHVE) scheme \cite{ShiW08} and Okamoto and Takashima constructed
an anonymous HIBE scheme with linear-size ciphertexts from a Hierarchical
Inner Product Encryption (HIPE) scheme \cite{OkamotoT09,LewkoOSTW10,
OkamotoT10,OkamotoT12}. However, currently known anonymous HIBE schemes from
PE schemes with the delegation capability only have linear-size ciphertexts.
It is also possible to derive anonymous HIBE from anonymous Spatial
Encryption (SE) \cite{BonehH08,ChenLLW11}. However, there is no known
anonymous SE scheme with constant-size ciphertexts. Thus the construction of
efficient and fully secure anonymous HIBE with short ciphertexts is an
unsolved problem.

\subsection{Our Contributions}

Motivated by the above challenge, we propose the first fully secure and
anonymous HIBE scheme with short ciphertexts in prime order (asymmetric)
bilinear groups. The comparison between previous HIBE schemes and ours is
given in Table~\ref{tab:comp-hibe}. To construct a fully secure and anonymous
HIBE scheme, we use the IBE scheme in prime order (asymmetric) bilinear
groups of Lewko and Waters \cite{LewkoW10}. Note that their IBE scheme does
not even converted to a (non-anonymous) HIBE scheme with short ciphertexts
since it does not support private key re-randomization\footnote{To support
private key re-randomization using a public key, some elements $\hat{g},
\hat{u}, \hat{h} \in \hat{G}$ in a private key should be moved to a public
key. However, these elements cannot be moved to the public key since the
proof of dual system encryption goes wrong.}.

\begin{table*}[t]
\caption{Comparison between previous HIBE schemes and ours}
\label{tab:comp-hibe} \vs \small \centerline{\tabcolsep=9.1pt
\renewcommand{\arraystretch}{1.3}
\begin{tabular}[c]{llllllll}
\hline
    Scheme  & ANON & R.L. & Prime & PP Size & SK Size & CT Size & Assumption \\
\hline
    GS-HIBE \cite{GentryS02} & No & $\Omega(q^l)$ & Yes & $O(\lambda)$ &
            $O(l \lambda)$ & $O(l \lambda)$ & BDH (ROM) \\
    BB-HIBE \cite{BonehB04e} & No & $\Omega(2^{\lambda l})$ & Yes & $O(l \lambda)$ &
            $O(l \lambda)$ & $O(l \lambda)$ & DBDH \\
    BBG-HIBE \cite{BonehBG05} & No & $\Omega(2^{\lambda l})$ & Yes & $O(l \lambda)$ &
            $O(l \lambda)$ & $2 k + k_T$ & $q$-Type \\
    CS-HIBE \cite{ChatterjeeS06} & No & $\Omega(q^l)$ & Yes & $O(l \lambda)$ &
            $O(l \lambda)$ & $O(l \lambda)$ & DBDH \\
\cdashline{3-8}
    Waters-HIBE \cite{Waters09} & No & $\Omega(q^2)$ & Yes & $O(l \lambda)$ &
            $O(l \lambda)$ & $O(l \lambda)$ & DBDH, DLIN \\
    LW-HIBE \cite{LewkoW10} & No  & $\Omega(q)$ & No & $O(l \lambda)$ &
            $O(l \lambda)$ & $2 k + k_T$ & Static \\
    LW-HIBE \cite{LewkoW11} & No  & $\Omega(q)$ & No & $O(\lambda)$ &
            $O(l \lambda)$ & $O(l \lambda)$ & Static \\
    OT-HIPE \cite{OkamotoT11} & No  & $\Omega(q)$ & Yes & $O(l^4 \lambda)$ &
            $O(l^2 \lambda)$ & $133 k + k_T$ & DLIN \\
    Lewko-HIBE \cite{Lewko12} & No  & $\Omega(q)$ & Yes & $O(\lambda)$ &
            $O(l \lambda)$ & $O(l \lambda)$ & DLIN \\
\cdashline{2-8}
    BW-HIBE \cite{BoyenW06} & Yes & $\Omega(2^{\lambda l})$ & Yes & $O(l^2 \lambda)$ &
            $O(l^2 \lambda)$ & $O(l \lambda)$ & DBDH, DLIN \\
    SKOS-HIBE \cite{SeoKOS09} & Yes & $\Omega(2^{\lambda l})$ & No & $O(l \lambda)$ &
            $O(l \lambda)$ & $3 k + k_T$ & $q$-Type \\
    Ducas-HIBE \cite{Ducas10} & Yes & $\Omega(2^{\lambda l})$ & Yes & $O(l \lambda)$ &
            $O(l \lambda)$ & $3 k + k_T$ & $q$-Type \\
    LL-HIBE \cite{LeeL10} & Yes & $\Omega(2^{\lambda l})$ & Yes & $O(l \lambda)$ &
            $O(l \lambda)$ & $6 k + k_T$ & $q$-Type \\
\cdashline{3-8}
    DIP-HIBE \cite{CaroIP10} & Yes & $\Omega(q)$ & No & $O(l \lambda)$ &
            $O(l \lambda)$ & $2 k + k_T$ & Static \\
    LOSTW-HIPE \cite{LewkoOSTW10} & Yes & $\Omega(lq)$ & Yes & $O(l^4 \lambda)$ &
            $O(l^3 \lambda)$ & $O(l^2 \lambda)$ & $q$-Type \\
    OT-HIPE \cite{OkamotoT10} & Yes & $\Omega(l^2 q)$ & Yes & $O(l^3 \lambda)$ &
            $O(l^4 \lambda)$ & $O(l^2 \lambda)$ & DLIN \\
    OT-HIPE \cite{OkamotoT12} & Yes & $\Omega(lq)$ & Yes & $O(l^2 \lambda)$ &
            $O(l^2 \lambda)$ & $O(l \lambda)$ & DLIN \\
    Ours    & Yes & $\Omega(q)$ & Yes & $O(l \lambda)$ & $O(l \lambda)$ &
            $6 k + k_T$ & Static \\
\hline
\multicolumn{8}{l}{ANON = anonymity, R.L. = reduction loss, Prime = prime order
    bilinear groups} \\
\multicolumn{8}{l}{$\lambda$ = security parameter, $l$ = hierarchical depth,
    $q$ = polynomial value, $k, k_T$ = the bit size of group $\G$ and
    $\G_T$}
\end{tabular}
}
\end{table*}

To construct an anonymous HIBE scheme, we should devise techniques for
private key re-randomization and ciphertext anonymization. The private key
re-randomization process is required in the delegation algorithm of HIBE and
anonymous HIBE. In HIBE, private keys are simply re-randomized using the
public elements of public parameters. However, private keys of anonymous HIBE
cannot be simply re-randomized using the public elements because an attacker
can break anonymity using the public elements. To solve this problem, we may
use the \textit{private re-randomization} technique of Boyen and Waters
\cite{BoyenW06} that re-randomizes private keys using the private elements of
private keys. Nevertheless, if the private re-randomization technique is used
in the dual system encryption method, then additional random values in
semi-functional private keys are not completely randomized in the proof that
distinguishes a normal private key from a semi-functional private key.

To resolve this difficulty, we define two types of semi-functional private
keys as semi-functional type-1 and semi-functional type-2, and we show that
it is hard to distinguish these two types of semi-functional private keys.
The main idea to provide ciphertext anonymity is that the Decisional
Diffie-Hellman (DDH) assumption still holds in asymmetric bilinear groups of
prime order. We prove the anonymity property of our scheme by introducing a
new assumption since the simple DDH assumption is not enough for the security
proof. Furthermore, we implemented our anonymous HIBE scheme using the PBC
library to support our claim of efficiency and we measured the performance of
our scheme.

\subsection{Related Work}

IBE was introduced to solve the certificate management problem in public key
encryption systems, but it additionally requires a Private-Key Generator
(PKG) \cite{BonehF01,BonehF03}. HIBE was invented to reduce the burden of the
IBE's PKG by re-arranging an identity as a hierarchical tree structure and by
allowing the delegation of private key generation from upper level users to
lower level users \cite{HorwitzL02}. Gentry and Silverberg proposed the first
HIBE scheme in the random oracle model \cite{GentryS02}. Canetti et al.
constructed the first HIBE scheme without random oracles and introduced a
selective model to prove the security of their scheme \cite{CanettiHK03}. The
selective model was widely used in the security proof of IBE and HIBE even
though it is weaker than the full model. For instance, Boneh and Boyen
proposed an efficient HIBE scheme with linear-size ciphertexts
\cite{BonehB04e,BonehB11}, and Boneh et al. proposed an HIBE scheme with
constant-size ciphertexts \cite{BonehBG05}.

To construct a fully secure HIBE scheme, Boneh and Boyen showed that a
selectively secure HIBE scheme is naturally converted to a fully secure HIBE
scheme with exponential loss of a reduction efficiency \cite{BonehB04e}.
However, this approach has a serious problem -- that is, the efficiency of
the reduction is $1/\Omega(2^{\lambda l})$ where $\lambda$ is a security
parameter and $l$ is the maximum hierarchical depth. To remedy this
situation, Waters proposed an HIBE scheme by extending his fully secure IBE
scheme with an efficient reduction to a HIBE scheme \cite{Waters05}, and
Chatterjee and Sarkar improved the efficiency of Waters' scheme
\cite{ChatterjeeS06}. However, these schemes also have the problem of an
inefficient reduction $1/\Omega(q^l)$ in the hierarchical setting where $q$
is a polynomial value. Gentry and Halevi proposed another fully secure HIBE
scheme with an efficient reduction by using complex assumptions
\cite{GentryH09}. Recently, Waters introduced the dual system encryption
method that can be used to construct a fully secure HIBE scheme with an
efficient reduction under simple assumptions \cite{Waters09,LewkoW10}.

Anonymous IBE is related to public key encryption with keyword search (PEKS)
\cite{BonehCOP04,Gentry06}, and the concept of anonymous HIBE was introduced
by Abdalla et al. \cite{AbdallaBC+05} by extending the concept of anonymous
IBE. Boyen and Waters proposed the first anonymous HIBE scheme without random
oracles and proved its security in the selective model \cite{BoyenW06}. For
the construction of anonymous HIBE, they devised a linear splitting technique
for ciphertext anonymity and a private re-randomization technique for private
key randomization. Seo et al. proposed the first anonymous HIBE scheme with
short ciphertexts in composite order bilinear groups \cite{SeoKOS09}. Ducas
constructed anonymous HIBE schemes using asymmetric bilinear groups of prime
order \cite{Ducas10}. Lee and Lee proposed an efficient anonymous HIBE scheme
with short ciphertexts that is secure in all types of bilinear groups of
prime order \cite{LeeL10}. De Caro et al. proposed the first fully secure and
anonymous HIBE scheme with short ciphertexts using the dual system encryption
method in composite order bilinear groups \cite{CaroIP10}.

HIBE schemes also can be constructed from Attribute Based Encryption (ABE)
schemes \cite{GoyalPSW06} and Predicate Encryption (PE) schemes with
delegation capabilities \cite{ShiW08,OkamotoT09}. PE schemes with linear-size
ciphertexts that have the delegation capability include the dHVE scheme of
Shi and Waters in composite order bilinear groups \cite{ShiW08} and HIPE
schemes of Okamoto and Takashima based on dual pairing vector spaces
\cite{OkamotoT09,LewkoOSTW10,OkamotoT10,OkamotoT12}. A non-anonymous HIPE
scheme based on dual pairing vector spaces can have constant-size
ciphertexts, but the ciphertext should contain a linear-size identity vector
\cite{OkamotoT11}. Though bilinear groups were widely used in the
construction of HIBE, some HIBE schemes were designed in lattices
\cite{CashHKP10,AgrawalBB10e,AgrawalBB10l}.

\section{Preliminaries}

We define anonymous HIBE and give the formal definition of its full model
security. Let $\mathcal{I}$ be an identity space and $\mathcal{M}$ be a
message space. A hierarchical identity $ID$ of depth $c$ is defined as an
identity vector $(I_1, \ldots, I_c) \in \mathcal{I}^c$. A hierarchical
identity $ID = (I_1, \ldots, I_c)$ of depth $c$ is a prefix of a hierarchical
identity $ID' = (I'_1, \ldots, I'_d)$ of depth $d$ if $c \leq d$ and for all
$i \in \{1, \ldots, c\}$, $I_i = I'_i$.

\subsection{Anonymous HIBE}

An anonymous HIBE scheme consists of five algorithms (\textbf{Setup, KeyGen,
Delegate, Encrypt, Decrypt}). Formally it is defined as:
\begin{description}
\item \textbf{Setup}($1^{\lambda}, l$). The setup algorithm takes as
    input a security parameter $1^{\lambda}$ and a maximum hierarchical
    depth $l$. It outputs a master key $MK$ and public parameters $PP$.

\item \textbf{KeyGen}($ID, MK, PP$). The key generation algorithm takes
    as input a hierarchical identity $ID$ of depth $m$ where $m \leq l$,
    the master key $MK$, and the public parameters $PP$. It outputs a
    private key $SK_{ID}$ for $ID$.

\item \textbf{Delegate}($ID', SK_{ID}, PP$). The delegation algorithm
    takes as input a hierarchical identity $ID'$ of depth $m+1$ where
    $m+1 \leq l$, a private key $SK_{ID}$ for a hierarchical identity
    $ID$ of depth $m$, and the public parameters $PP$. If $ID$ is a
    prefix of $ID'$, then it outputs a delegated private key $SK_{ID'}$
    for $ID'$.

\item \textbf{Encrypt}($ID, M, PP$). The encryption algorithm takes as
    input a hierarchical identity $ID$ of depth $n$ where $n \leq l$, a
    message $M \in \mathcal{M}$, and the public parameters $PP$. It
    outputs a ciphertext $CT$ for $ID$ and $M$.

\item \textbf{Decrypt}($CT, SK_{ID}, PP$). The decryption algorithm takes
    as input a ciphertext $CT$ for a hierarchical identity $ID'$, a
    private key $SK_{ID}$ for a hierarchical identity $ID$, and the
    public parameters $PP$. If $ID = ID'$, then it outputs an encrypted
    message $M$.
\end{description}

The correctness property of anonymous HIBE is defined as follows: For all
$MK, PP$ generated by $\textbf{Setup}$, all $ID, ID' \in \mathcal{I}^n$, any
$SK_{ID}$ generated by $\textbf{KeyGen}$, and any $M$, it is required that
\begin{itemize}
\item If $ID = ID'$, then $\textbf{Decrypt}(\textbf{Encrypt}(ID', M, PP),
    SK_{ID}, PP) = M$.
\item If $ID \neq ID'$, then $\textbf{Decrypt}(\textbf{Encrypt}(ID', M,
    PP), SK_{ID}, PP) = \perp$ with all but negligible probability.
\end{itemize}
The second condition of the correctness property is not a trivial one to
satisfy since the decryption algorithm of anonymous HIBE cannot easily check
whether $ID = ID'$ or not because of anonymity. One possible relaxation is to
use a computational condition instead of a statistical condition. For a
computational condition, we can use weak robustness of Abdalla et al.
\cite{AbdallaBN10}.

The security property of anonymous HIBE under a chosen plaintext attack is
defined in terms of the following experiment between a challenger $\mcC$ and
a PPT adversary $\mcA$:
\begin{enumerate}
\item \textbf{Setup}: $\mcC$ runs $\textbf{Setup}(1^{\lambda}, l)$ to
    generate a master key $MK$ and public parameters $PP$. It keeps $MK$
    to itself and gives $PP$ to $\mcA$.

\item \textbf{Query 1}: $\mcA$ may adaptively request a polynomial number
    of private keys for hierarchical identities $ID_1, \ldots, ID_{q_1}$
    of arbitrary depths. In response, $\mcC$ gives the corresponding
    private keys $SK_{ID_1}, \ldots, SK_{ID_{q_1}}$ to $\mcA$ by running
    $\textbf{KeyGen}(ID_i, MK, PP)$.

\item \textbf{Challenge}: $\mcA$ submits two hierarchical identities
    $ID_0^*, ID_1^* \in \mathcal{I}^n$ and two messages $M_0^*, M_1^*$
    with equal length subject to the restriction: for all $ID_i$ of
    private key queries, $ID_i$ is not a prefix of $ID_0^*$ and $ID_1^*$.
    $\mcC$ flips a random coin $\gamma \in \{0,1\}$ and gives the
    challenge ciphertext $CT^*$ to $\mcA$ by running
    $\textbf{Encrypt}(ID_{\gamma}^*, M_{\gamma}^*, PP)$.

\item \textbf{Query 2}: $\mcA$ may continue to request a polynomial
    number of private keys for hierarchical identities $ID_{q_1 +1},
    \ldots, ID_q$ subject to the restriction as before.

\item \textbf{Guess}: $\mcA$ outputs a guess $\gamma' \in \{0,1\}$ of
    $\gamma$, and wins the game if $\gamma' = \gamma$.
\end{enumerate}
The advantage of $\mcA$ is defined as $\Adv_{\mcA}^{AHIBE}(\lambda) = \big|
\Pr[\gamma = \gamma'] - 1/2 \big|$ where the probability is taken over all
the randomness of the experiment. An anonymous HIBE scheme is fully secure
under a chosen plaintext attack if for all PPT adversary $\mcA$, the
advantage of $\mcA$ in the above experiment is negligible in the security
parameter $\lambda$.

The security experiment of anonymous HIBE can be relaxed to complete one
introduced by Shi and Waters \cite{ShiW08} that traces the path of
delegation. Our definition of the security experiment that does not trace the
path of delegation is stronger than the complete one of Shi and Waters. Thus
if an anonymous HIBE scheme is secure in the security experiment of this
section, then the scheme is also secure in the complete one.

\subsection{Asymmetric Bilinear Groups}

Let $\G, \hat{\G}$ and $\G_{T}$ be multiplicative cyclic groups of prime
order $p$ with the security parameter $\lambda$. Let $g, \hat{g}$ be
generators of $\G, \hat{\G}$. The bilinear map $e:\G \times \hat{\G}
\rightarrow \G_{T}$ has the following properties:
\begin{enumerate}
\item Bilinearity: $\forall u \in \G, \forall \hat{v} \in \hat{\G}$ and
    $\forall a,b \in \Z_p$, $e(u^a,\hat{v}^b) = e(u,\hat{v})^{ab}$.
\item Non-degeneracy: $\exists g, \hat{g}$ such that $e(g,\hat{g})$ has
    order $p$, that is, $e(g,\hat{g})$ is a generator of $\G_T$.
\end{enumerate}
We say that $\G, \hat{\G}, \G_T$ are bilinear groups with no efficiently
computable isomorphisms if the group operations in $\G, \hat{\G},$ and $\G_T$
as well as the bilinear map $e$ are all efficiently computable, but there are
no efficiently computable isomorphisms between $\G$ and $\hat{\G}$.

\subsection{Complexity Assumptions}

We introduce five assumptions under asymmetric bilinear groups of prime
order. Assumptions 1 and 2 were introduced in Lewko and Waters
\cite{LewkoW10}, and Assumptions 3 and 4 are well-known. Assumption 5
(Asymmetric 3-Party Diffie-Hellman) is an asymmetric version of the Composite
3-Party Diffie-Hellman assumption introduced by Boneh and Waters
\cite{BonehW07} with a slight modification by augmenting one additional
element, and it is secure in the generic group model.

\vs \noindent \textbf{Assumption 1 (LW1)} Let $(p, \G, \hat{\G}, \G_T, e)$ be
a description of the asymmetric bilinear group of prime order $p$ with the
security parameter $\lambda$. Let $g, \hat{g}$ be generators of $\G,
\hat{\G}$ respectively. The assumption is that if the challenge values
    \begin{align*}
    D = ((p, \G, \hat{\G}, \G_T, e),
        g, g^a, g^b, g^{ab^2}, g^{b^2}, g^{b^3},
        g^c, g^{ac}, g^{bc}, g^{b^2 c}, g^{b^3 c}, \hat{g}, \hat{g}^b)
        \mbox{ and } T
    \end{align*}
are given, no PPT algorithm $\mcB$ can distinguish $T = T_0 = g^{ab^2 c}$
from $T = T_1 = g^d$ with more than a negligible advantage. The advantage of
$\mcB$ is defined as $\Adv_{\mcB}^{A1}(\lambda) = \big| \Pr[\mcB(D,T_0)=0] -
\Pr[\mcB(D,T_1)=0] \big|$ where the probability is taken over the random
choice of $a, b, c, d \in \Z_p$.

\vs \noindent \textbf{Assumption 2 (LW2)} Let $(p, \G, \hat{\G}, \G_T, e)$ be
a description of the asymmetric bilinear group of prime order $p$ with the
security parameter $\lambda$. Let $g, \hat{g}$ be generators of $\G,
\hat{\G}$ respectively. The assumption is that if the challenge values
    \begin{align*}
    D = ((p, \G, \hat{\G}, \G_T, e),
        g, g^a, g^{a^2}, g^{bx}, g^{abx}, g^{a^2x},
        \hat{g}, \hat{g}^a, \hat{g}^b, \hat{g}^c) \mbox{ and } T
    \end{align*}
are given, no PPT algorithm $\mcB$ can distinguish $T = T_0 = \hat{g}^{bc}$
from $T = T_1 = \hat{g}^d$ with more than a negligible advantage. The
advantage of $\mcB$ is defined as $\Adv_{\mcB}^{A2}(\lambda) = \big|
\Pr[\mcB(D,T_0)=0] - \Pr[\mcB(D,T_1)=0] \big|$ where the probability is taken
over the random choice of $a, b, c, x, d \in \Z_p$.

\vs \noindent \textbf{Assumption 3 (Symmetric eXternal Diffie-Hellman)} Let
$(p, \G, \hat{\G}, \G_T, e)$ be a description of the asymmetric bilinear
group of prime order $p$ with the security parameter $\lambda$. Let $g,
\hat{g}$ be generators of $\G, \hat{\G}$ respectively. The assumption is that
if the challenge values
    \begin{align*}
    D = ((p, \G, \hat{\G}, \G_T, e),
        g, \hat{g}, \hat{g}^a, \hat{g}^b) \mbox{ and } T
    \end{align*}
are given, no PPT algorithm $\mcB$ can distinguish $T = T_0 = \hat{g}^{ab}$
from $T = T_1 = \hat{g}^c$ with more than a negligible advantage. The
advantage of $\mcB$ is defined as $\Adv_{\mcB}^{A3}(\lambda) = \big|
\Pr[\mcB(D,T_0)=0] - \Pr[\mcB(D,T_1)=0] \big|$ where the probability is taken
over the random choice of $a, b, c \in \Z_p$.

\vs \noindent \textbf{Assumption 4 (Decisional Bilinear Diffie-Hellman)} Let
$(p, \G, \hat{\G}, \G_T, e)$ be a description of the asymmetric bilinear
group of prime order $p$ with the security parameter $\lambda$. Let $g,
\hat{g}$ be generators of $\G, \hat{\G}$ respectively. The assumption is that
if the challenge values
    \begin{align*}
    D = ((p, \G, \hat{\G}, \G_T, e),~
        g, g^a, g^b, g^c, \hat{g}, \hat{g}^a, \hat{g}^b, \hat{g}^c)
        \mbox{ and } T
    \end{align*}
are given, no PPT algorithm $\mcB$ can distinguish $T = T_0 = e(g,
\hat{g})^{abc}$ from $T = T_1 = e(g, \hat{g})^d$ with more than a negligible
advantage. The advantage of $\mcB$ is defined as $\Adv_{\mcB}^{A4}(\lambda) =
\big| \Pr[\mcB(D,T_0)=0] - \Pr[\mcB(D,T_1)=0] \big|$ where the probability is
taken over the random choice of $a, b, c, d \in \Z_p$.

\vs \noindent \textbf{Assumption 5 (Asymmetric 3-Party Diffie-Hellman)} Let
$(p, \G, \hat{\G}, \G_T, e)$ be a description of the asymmetric bilinear
group of prime order $p$ with the security parameter $\lambda$. Let $g,
\hat{g}$ be generators of $\G, \hat{\G}$ respectively. The assumption is that
if the challenge values
    \begin{align*}
    D = ( &(p, \G, \hat{\G}, \G_T, e),~
        g, g^a, g^b, g^c, g^{ab}, g^{a^2 b}, \hat{g}, \hat{g}^a, \hat{g}^b)
        \mbox{ and } T
    \end{align*}
are given, no PPT algorithm $\mcB$ can distinguish $T = T_0 = g^{abc}$ from
$T = T_1 = g^d$ with more than a negligible advantage. The advantage of
$\mcB$ is defined as $\Adv_{\mcB}^{A5}(\lambda) = \big| \Pr[\mcB(D,T_0)=0] -
\Pr[\mcB(D,T_1)=0] \big|$ where the probability is taken over the random
choice of $a, b, c, d \in \Z_p$.

\section{Anonymous HIBE}

We construct an anonymous HIBE scheme in prime order (asymmetric) bilinear
groups and prove its full model security under static assumptions.

\subsection{Construction}

Let $\mathcal{I} = \Z_p^*$. Our anonymous HIBE scheme is described as
follows:

\begin{description}
\item [\textbf{Setup}($1^{\lambda}, l$):] This algorithm first generates
    the asymmetric bilinear groups $\G, \hat{\G}, \G_T$ of prime order
    $p$ of bit size $\Theta(\lambda)$. It chooses random elements $g \in
    \G$ and $\hat{g} \in \hat{\G}$. It also chooses random exponents
    $\nu, \phi_1, \phi_2 \in \Z_p$ and sets $\tau = \phi_1 + \nu \phi_2$.
    Next, it selects random exponents $y_h, \{ y_{u_i} \}_{i=1}^l, y_w,
    \alpha \in \Z_p$ and sets $h = g^{y_h}, \hat{h} = \hat{g}^{y_h}, \{
    u_i = g^{y_{u_i}}, \hat{u}_i = \hat{g}^{y_{u_i}} \}_{i=1}^l, \hat{w}
    = \hat{g}^{y_w}$. It outputs a master key $MK = (\hat{g},
    \hat{g}^{\alpha}, \hat{h}, \{ \hat{u}_i \}_{i=1}^l )$ and public
    parameters as
    \begin{align*}
    PP = \Big(~
        g, g^{\nu}, g^{-\tau},~ h, h^{\nu}, h^{-\tau},~
        \{ u_i, u_i^{\nu}, u_i^{-\tau} \}_{i=1}^l,~
        \hat{w}^{\phi_1}, \hat{w}^{\phi_2}, \hat{w},~
        \Omega = e(g, \hat{g})^{\alpha}
    ~\Big).
    \end{align*}

\item [\textbf{KeyGen}($ID, MK, PP$):] This algorithm takes as input a
    hierarchical identity $ID = (I_1, \ldots, I_m) \in \mathcal{I}^m$ and
    the master key $MK$. It first selects random exponents $r_1, c_1,
    c_2, \{ c_{3,i} \}_{i=m+1}^l \in \Z_p$ and creates the decryption and
    delegation components of a private key as
    \begin{align*}
    &   K_{1,1} = \hat{g}^{\alpha} (\hat{h} \prod_{i=1}^m \hat{u}_i^{I_i})^{r_1}
                  (\hat{w}^{\phi_1})^{c_1},~
        K_{1,2} = (\hat{w}^{\phi_2})^{c_1},~
        K_{1,3} = \hat{w}^{c_1}, \\
    &   K_{2,1} = \hat{g}^{r_1} (\hat{w}^{\phi_1})^{c_2},~
        K_{2,2} = (\hat{w}^{\phi_2})^{c_2},~
        K_{2,3} = \hat{w}^{c_2}, \\
    &   \big\{
        L_{3,i,1} = \hat{u}_i^{r_1} (\hat{w}^{\phi_1})^{c_{3,i}},~
        L_{3,i,2} = (\hat{w}^{\phi_2})^{c_{3,i}},~
        L_{3,i,3} = \hat{w}^{c_{3,i}}
        \big\}_{i=m+1}^l.
    \end{align*}
    Next, it selects random exponents $r_2, c_4, c_5, \{ c_{6,i}
    \}_{i=m+1}^l \in \Z_p$ and creates the randomization components of
    the private key as
    \begin{align*}
    &   R_{1,1} = (\hat{h} \prod_{i=1}^m \hat{u}_i^{I_i})^{r_2}
                  (\hat{w}^{\phi_1})^{c_4},~
        R_{1,2} = (\hat{w}^{\phi_2})^{c_4},~
        R_{1,3} = \hat{w}^{c_4},~ \\
    &   R_{2,1} = \hat{g}^{r_2} (\hat{w}^{\phi_1})^{c_5},~
        R_{2,2} = (\hat{w}^{\phi_2})^{c_5},~
        R_{2,3} = \hat{w}^{c_5},~ \\
    &   \big\{
        R_{3,i,1} = \hat{u}_i^{r_2} (\hat{w}^{\phi_1})^{c_{6,i}},~
        R_{3,i,2} = (\hat{w}^{\phi_2})^{c_{6,i}},~
        R_{3,i,3} = \hat{w}^{c_{6,i}}
        \big\}_{i=m+1}^l.
    \end{align*}
    Finally, it outputs a private key as
    \begin{align*}
    SK_{ID} = \Big(~
    &   K_{1,1}, K_{1,2}, K_{1,3},~
        K_{2,1}, K_{2,2}, K_{2,3},~
        \{ L_{3,i,1}, L_{3,i,2}, L_{3,i,3} \}_{i=m+1}^l,~ \\
    &   R_{1,1}, R_{1,2}, R_{1,3},~
        R_{2,1}, R_{2,2}, R_{2,3},~
        \{ R_{3,i,1}, R_{3,i,2}, R_{3,i,3} \}_{i=m+1}^l
    ~\Big).
    \end{align*}

\item [\textbf{Delegate}($ID', SK_{ID}, PP$):] This algorithm takes as
    input a hierarchical identity $ID' = (I_1, \ldots, I_{m+1}) \in
    \mathcal{I}^{m+1}$ and a private key $SK_{ID}$ for a hierarchical
    identity $ID = (I_1, \ldots, I_m) \in \mathcal{I}^m$ where $ID$ is a
    prefix of $ID'$. Let $(W_1, W_2, W_3) = (\hat{w}^{\phi_1},
    \hat{w}^{\phi_2}, \hat{w})$. It first selects random exponents
    $\gamma_1, \delta_1, \delta_2, \{ \delta_{3,i} \}_{i=m+2}^l \in \Z_p$
    and creates the decryption and delegation components of a delegated
    private key as
    \begin{align*}
    &   \big( K'_{1,k} = K_{1,k} L_{3,m+1,k}^{I_{m+1}} \cdot
            (R_{1,k} R_{3,m+1,k}^{I_{m+1}})^{\gamma_1} W_k^{\delta_1}
        \big)_{1 \leq k\leq 3},~
        \big( K'_{2,k} = K_{2,k} \cdot R_{2,k}^{\gamma_1} W_k^{\delta_2}
        \big)_{1 \leq k\leq 3},~ \\
    &   \big\{ \big(
            L'_{3,i,k} = L_{3,i,k} \cdot R_{3,i,k}^{\gamma_1} W_k^{\delta_{3,i}}
        \big)_{1 \leq k \leq 3} \big\}_{i=m+2}^l.
    \end{align*}
    Next, it selects random exponents $\gamma_2, \delta_4, \delta_5, \{
    \delta_{6,i} \}_{i=m+2}^l \in \Z_p$ and creates the randomization
    components of the delegated private key as
    \begin{align*}
    &   \big( R'_{1,k} = (R_{1,k} R_{3,m+1,k}^{I_{m+1}})^{\gamma_2}
            W_k^{\delta_4} \big)_{1 \leq k \leq 3},~
        \big( R'_{2,k} = R_{2,k}^{\gamma_2} W_k^{\delta_5}
            \big)_{1 \leq k \leq 3},~
        \big\{ \big(
            R'_{3,i,k} = R_{3,i,k}^{\gamma_2} W_k^{\delta_{6,i}}
        \big)_{1 \leq k \leq 3} \big\}_{i=m+2}^l.
    \end{align*}
    Finally, it outputs a delegated private key as
    \begin{align*}
    SK_{ID'} = \Big(~
    &   K'_{1,1}, K'_{1,2}, K'_{1,3},~
        K'_{2,1}, K'_{2,2}, K'_{2,3},~
        \{ L'_{3,i,1}, L'_{3,i,2}, L'_{3,i,3} \}_{i=m+2}^l,~ \\
    &   R'_{1,1}, R'_{1,2}, R'_{1,3},~
        R'_{2,1}, R'_{2,2}, R'_{2,3},~
        \{ R'_{3,i,1}, R'_{3,i,2}, R'_{3,i,3} \}_{i=m+2}^l
    ~\Big).
    \end{align*}
    The distribution of the delegated private key is the same as the
    original private key since the random values are defined as $r'_1 =
    r_1 + r_2 \gamma_1, r'_2 = r_2 \gamma_2$ where $r_1, r_2$ are random
    exponents in the private key $SK_{ID}$. Note that $c_1, c_2, \{
    c_{3,i} \}, c_4, c_5, \{ c_{6,i} \}$ are perfectly re-randomized
    since $\hat{w}^{\phi_1}, \hat{w}^{\phi_2}, \hat{w}$ are publicly
    known and $\delta_1, \delta_2, \{ \delta_{3,i} \}, \delta_4,
    \delta_5, \{ \delta_{6,i} \}$ are chosen randomly.

\item [\textbf{Encrypt}($ID, M, PP$):] This algorithm takes as input a
    hierarchical identity $ID = (I_1, \ldots, I_n) \in \mathcal{I}^n$, a
    message $M \in \G_T$, and the public parameter $PP$. It selects a
    random exponent $t \in \Z_p$ and outputs a ciphertext as
    \begin{align*}
    CT = \Big(~
    &   C       = \Omega^t M,~
        C_{1,1} = g^t,~
        C_{1,2} = (g^{\nu})^t,~
        C_{1,3} = (g^{-\tau})^t,~ \\
    &   C_{2,1} = (h \prod_{i=1}^n u_i^{I_i})^t,~
        C_{2,2} = (h^{\nu} \prod_{i=1}^n (u_i^{\nu})^{I_i})^t,~
        C_{2,3} = (h^{-\tau} \prod_{i=1}^n (u_i^{-\tau})^{I_i})^t
    ~\Big).
    \end{align*}

\item [\textbf{Decrypt}($CT, SK_{ID}, PP$):] This algorithm takes as
    input a ciphertext $CT$ and a private key $SK_{ID}$ for a
    hierarchical identity $ID = (I_1, \ldots, I_n)$. It outputs the
    encrypted message as
    \begin{align*}
    M \leftarrow C \cdot
        \prod_{i=1}^3 e(C_{1,i}, K_{1,i})^{-1} \cdot
        \prod_{i=1}^3 e(C_{2,i}, K_{2,i}).
    \end{align*}
\end{description}

\subsection{Correctness}

The first condition of the correctness property can be easily checked by the
following equation as
    \begin{eqnarray*}
    \prod_{i=1}^3 e(C_{1,i}, K_{1,i})^{-1} \cdot \prod_{i=1}^3 e(C_{2,i}, K_{2,i})
     =  e(g^t, \hat{g}^{\alpha} (\hat{h} \prod_{i=1}^n \hat{u}_i^{I_i})^{r_1})^{-1} \cdot
        e((h \prod_{i=1}^n u_i^{I_i})^t, \hat{g}^{r_1})
     =  e(g, \hat{g})^{-\alpha t}
    \end{eqnarray*}
since the inner product of $(1, \nu, -\tau)$ and $(\phi_1, \phi_2, 1)$ are
zero.
The second condition of the correctness property can be satisfied by using
the technique of Boneh and Waters \cite{BonehW07} that uses the limited
message space. If we use a computational condition instead of a statistical
condition, then we can achieve weak robustness by using the transformation of
Abdalla et al. \cite{AbdallaBN10}.

\subsection{Security Analysis}

\begin{theorem} \label{thm:ahibe-prime}
The above anonymous HIBE scheme is fully secure under a chosen plaintext
attack if Assumptions 1, 2, 3, 4 and 5 hold. That is, for any PPT adversary
$\mcA$, there exist PPT algorithms $\mcB_1, \mcB_2, \mcB_3, \mcB_4$, and
$\mcB_5$ such that
    \begin{eqnarray*}
    \Adv_{\mcA}^{AHIBE}(\lambda)
    \leq \Adv_{\mcB_1}^{A1}(\lambda) +
        q \big(
            \Adv_{\mcB_2}^{A2}(\lambda) +
            \Adv_{\mcB_3}^{A3}(\lambda) \big)
        + \Adv_{\mcB_4}^{A4}(\lambda) +
        \Adv_{\mcB_5}^{A5}(\lambda).
    \end{eqnarray*}
where $q$ is the maximum number of private key queries of $\mcA$.
\end{theorem}

\begin{proof}
To prove the security of our scheme, we use the dual system encryption
technique of \cite{Waters09,LewkoW10}. We first describe a semi-functional
key generation algorithm and a semi-functional encryption algorithm. They are
not used in a real system, but they are used in the security proof. For
semi-functionality, we set $f = g^{y_f}, \hat{f} = \hat{g}^{y_f}$ where $y_f$
is a random exponent in $\Z_p$.

\begin{description}
\item [\textbf{KeyGenSF-1}.] The semi-functional type-1 key generation
    algorithm first creates a normal private key using the master key.
    Let $(K'_{1,1}, \ldots, \{ R'_{3,i,1}, \ldots, R'_{3,i,3}
    \}_{i=m+1}^l)$ be the normal private key of a hierarchical identity
    $ID = (I_1, \ldots, I_m)$ with random exponents $r_1, r_2, c_1, c_2,
    \{ c_{3,i} \}, c_4, c_5, \{ c_{6,i} \} \in \Z_p$. It selects random
    exponents $s_{k,1}, z_{k,1}, \{ z_{k,2,i} \}_{i=m+1}^l, s_{k,2} \in
    \Z_p$ and outputs a semi-functional type-1 private key as
    \begin{align*}
    &   K_{1,1} = K'_{1,1} (\hat{f}^{-\nu})^{s_{k,1} z_{k,1}},~
        K_{1,2} = K'_{1,2} \hat{f}^{s_{k,1} z_{k,1}},~
        K_{1,3} = K'_{1,3}, \\
    &   K_{2,1} = K'_{2,1} (\hat{f}^{-\nu})^{s_{k,1}},~
        K_{2,2} = K'_{2,2} \hat{f}^{s_{k,1}},~
        K_{2,3} = K'_{2,3}, \\
    &   \big\{
        L_{3,i,1} = L'_{3,i,1} (\hat{f}^{-\nu})^{s_{k,1} z_{k,2,i}},~
        L_{3,i,2} = L'_{3,i,2} \hat{f}^{s_{k,1} z_{k,2,i}},~
        L_{3,i,3} = L'_{3,i,3} \big\}_{i=m+1}^l,~ \\
    &   R_{1,1} = R'_{1,1} (\hat{f}^{-\nu})^{s_{k,2} z_{k,1}},~
        R_{1,2} = R'_{1,2} \hat{f}^{s_{k,2} z_{k,1}},~
        R_{1,3} = R'_{1,3}, \\
    &   R_{2,1} = R'_{2,1} (\hat{f}^{-\nu})^{s_{k,2}},~
        R_{2,2} = R'_{2,2} \hat{f}^{s_{k,2}},~
        R_{2,3} = R'_{2,3}, \\
    &   \big\{
        R_{3,i,1} = R'_{3,i,1} (\hat{f}^{-\nu})^{s_{k,2} z_{k,2,i}},~
        R_{3,i,2} = R'_{3,i,2} \hat{f}^{s_{k,2} z_{k,2,i}},~
        R_{3,i,3} = R'_{3,i,3} \big\}_{i=m+1}^l.
    \end{align*}
    Note that the randomization components should contain the
    semi-functional part since this semi-functional part enables the
    correct simulation of the security proof for anonymity.

\item [\textbf{KeyGenSF-2}.] The semi-functional type-2 key generation
    algorithm first creates a normal private key using the master key.
    Let $(K'_{1,1}, \ldots, \{ R'_{3,i,1}, \ldots, R'_{3,i,3}
    \}_{i=m+1}^l)$ be the normal private key of a hierarchical identity
    $ID = (I_1, \ldots, I_m)$. It selects random exponents $s_{k,1},
    z_{k,1}, \{ z_{k,2,i} \}_{i=m+1}^l, s_{k,2}, z_{k,3}, \{ z_{k,4,i}
    \}_{i=m+1}^l \in \Z_p$ and outputs a semi-functional type-2 private
    key the same as the semi-functional type-1 private key except that
    the randomization components are generated as
    \begin{align*}
    &   R_{1,1} = R'_{1,1} (\hat{f}^{-\nu})^{s_{k,2} z_{k,3}},~
        R_{1,2} = R'_{1,2} \hat{f}^{s_{k,2} z_{k,3}},~
        R_{1,3} = R'_{1,3}, \\
    &   R_{2,1} = R'_{2,1} (\hat{f}^{-\nu})^{s_{k,2}},~
        R_{2,2} = R'_{2,2} \hat{f}^{s_{k,2}},~
        R_{2,3} = R'_{2,3}, \\
    &   \big\{
        R_{3,i,1} = R'_{3,i,1} (\hat{f}^{-\nu})^{s_{k,2} z_{k,4,i}},~
        R_{3,i,2} = R'_{3,i,2} \hat{f}^{s_{k,2} z_{k,4,i}},~
        R_{3,i,3} = R'_{3,i,3}
        \big\}_{i=m+1}^l.
    \end{align*}
    Note that new random exponents $z_{k,3}, \{ z_{k,4,i} \}_{i=1}^l$ are
    chosen to generate the randomization components of the
    semi-functional type-2 private key, whereas the same exponents
    $z_{k,1}, \{ z_{k,2,i} \}_{i=1}^l$ of the decryption and delegation
    components are used to generate the randomization components in the
    semi-functional type-1 private key.

\item [\textbf{EncryptSF}.] The semi-functional encryption algorithm
    first creates a normal ciphertext using the public parameters. Let
    $(C', C'_{1,1}, \ldots, C'_{2,3})$ be the normal ciphertext. It
    selects random exponents $s_c, z_c \in \Z_p$ and outputs a
    semi-functional ciphertext as
    \begin{align*}
    &   C       = C',~
        C_{1,1} = C'_{1,1},~
        C_{1,2} = C'_{1,2} f^{s_c},~
        C_{1,3} = C'_{1,3} (f^{-\phi_2})^{s_c},~ \\
    &   C_{2,1} = C'_{2,1},~
        C_{2,2} = C'_{2,2} f^{s_c z_c},~
        C_{2,3} = C'_{2,3} (f^{-\phi_2})^{s_c z_c}.
    \end{align*}
\end{description}

If we decrypt a semi-functional ciphertext by using a semi-functional type-2
private key, then the decryption fails since an additional element
    $e(f,\hat{f})^{s_c ((s_{k,1} z_{k,1} + s_{k,2} z_{k,3} \gamma) -
    (s_{k,1} + s_{k,2} \gamma) z_c)}$
remains. Note that the decryption can be done after re-randomizing the
private key using a random exponent $\gamma$. If $(s_{k,1} z_{k,1} + s_{k,2}
z_{k,3} \gamma) = (s_{k,1} + s_{k,2} \gamma) z_c$, then the decryption
algorithm succeeds. However, the probability of this is negligible since
$s_{k,1}, s_{k,2}, z_{k,1}, z_{k,3}, z_c, \gamma$ are randomly chosen. In
case of the semi-functional type-1 private key, the additional random element
can be restated as
    $e(f, \hat{f})^{(s_{k,1} + s_{k,2} \gamma) s_c (z_{k,1} - z_c)}$.
If $z_{k,1} = z_c$, then the decryption algorithm succeeds. In this case, we
say that the private key is \textit{nominally} semi-functional type-1.

\vs The security proof consists of a sequence of games. The first game will
be the original security game and the last one will be a game such that the
adversary has no advantage. We define the games as follows:

\begin{description}
\item [\textbf{Game} $\textbf{G}_0$.] This game is the original security
    game. That is, the private keys and the challenge ciphertext are
    normal.

\item [\textbf{Game} $\textbf{G}_1$.] We first modify $\textbf{G}_0$ into
    a new game $\textbf{G}_1$. This game is almost identical to
    $\textbf{G}_0$ except that the challenge ciphertext is
    semi-functional.

\item [\textbf{Game} $\textbf{G}_2$.] Next, we modify $\textbf{G}_1$ into
    a game $\textbf{G}_2$. In this game, the private keys are
    semi-functional type-2 and the challenge ciphertext is
    semi-functional. Suppose that an adversary makes at most $q$ private
    key queries. For the security proof, we define a sequence of games
    $\textbf{G}_{1,0}, \ldots, \textbf{G}'_{1,k}, \textbf{G}_{1,k},
    \ldots, \textbf{G}_{1,q}$ where $\textbf{G}_{1,0} = \textbf{G}_1$. In
    $\textbf{G}'_{1,k}$ and $\textbf{G}_{1,k}$, a normal private key is
    given to the adversary for all $j$-th private key queries such that
    $j > k$ and a semi-functional type-2 private key is given to the
    adversary for all $j$-th private key queries such that $j < k$.
    However, for $k$-th private key query, a semi-functional type-1
    private key is given to the adversary in $\textbf{G}'_{1,k}$ where as
    a semi-functional type-2 private key is given in $\textbf{G}_{1,k}$.
    It is obvious that $\textbf{G}_{1,q}$ is equal to $\textbf{G}_2$.

\item [\textbf{Game} $\textbf{G}_3$.] We now define a new game. This game
    differs from $\textbf{G}_2$ where the challenge ciphertext component
    $C$ is replaced by a random element in $\G_T$.

\item [\textbf{Game} $\textbf{G}_4$.] Finally, we change $\textbf{G}_3$
    to a new game $\textbf{G}_4$. In this game, the semi-functional
    ciphertext components $(C_{2,1}, C_{2,2}, C_{2,3})$ are formed as
    $(P^t, (P^{\nu})^t f^{s_c z_{c}}, (P^{-\tau})^t (f^{-\phi_2})^{s_c
    z_{c}})$  where $P$ is a random element in $\G$. In this game, the
    challenge ciphertext gives no information about the random coin
    $\gamma$. Therefore, the adversary can win this game with probability
    at most $1/2$.
\end{description}
Let $\Adv_{\mcA}^{G_j}$ be the advantage of $\mcA$ in $\textbf{G}_j$ for
$j=0, \ldots, 4$. Let $\Adv_{\mcA}^{G_{1,k}}$ and $\Adv_{\mcA}^{G'_{1,k}}$ be
the advantage of $\mcA$ in $\textbf{G}_{1,k}$ and $\textbf{G}'_{1,k}$ for
$k=0, \ldots, q$. It is clear that
    $\Adv_{\mcA}^{AHIBE}(\lambda) = \Adv_{\mcA}^{G_0}$,
    $\Adv_{\mcA}^{G_{1,0}} = \Adv_{\mcA}^{G_1}$,
    $\Adv_{\mcA}^{G_{1,q}} = \Adv_{\mcA}^{G_2}$,
    and $\Adv_{\mcA}^{G_4} = 0$.
From the following five Lemmas, we obtain that it is hard to distinguish
$\textbf{G}_{i-1}$ from $\textbf{G}_i$ under the given assumptions.
Therefore, we have that
    \begin{align*}
    \Adv_{\mcA}^{AHIBE}(\lambda)
    & = \Adv_{\mcA}^{G_0} +
        \sum_{i=1}^3 \big( \Adv_{\mcA}^{G_i} - \Adv_{\mcA}^{G_i} \big) -
        \Adv_{\mcA}^{G_4}
    \leq \sum_{i=1}^{4} \big| \Adv_{\mcA}^{G_{i-1}} - \Adv_{\mcA}^{G_i} \big| \\
    & = \Adv_{\mcB_1}^{A1}(\lambda) +
        \sum_{k=1}^q \big(
            \Adv_{\mcB_2}^{A2}(\lambda) +
            \Adv_{\mcB_3}^{A3}(\lambda) \big) +
        \Adv_{\mcB_4}^{A4}(\lambda) + \Adv_{\mcB_5}^{A5}(\lambda).
    \end{align*}
This completes our proof of Theorem~\ref{thm:ahibe-prime}.
\end{proof}

\begin{lemma} \label{lem:ahibe-prime-1}
If Assumption 1 holds, then no PPT algorithm can distinguish between
$\textbf{G}_0$ and $\textbf{G}_1$ with a non-negligible advantage. That is,
for any adversary $\mcA$, there exists a PPT algorithm $\mcB_1$ such that
    $\big| \Adv_{\mcA}^{G_0} - \Adv_{\mcA}^{G_1} \big| =
    \Adv_{\mcB_1}^{A1}(\lambda)$.
\end{lemma}

\begin{lemma} \label{lem:ahibe-prime-2}
If Assumption 2 holds, then no PPT algorithm can distinguish between
$\textbf{G}_{1,k-1}$ and $\textbf{G}'_{1,k}$ with a non-negligible advantage.
That is, for any adversary $\mcA$, there exists a PPT algorithm $\mcB_2$ such
that
    $\big| \Adv_{\mcA}^{G_{1,k-1}} - \Adv_{\mcA}^{G'_{1,k}} \big| =
    \Adv_{\mcB_2}^{A2}(\lambda)$.
\end{lemma}

\begin{lemma} \label{lem:ahibe-prime-3}
If Assumption 3 holds, then no PPT algorithm can distinguish between
$\textbf{G}'_{1,k}$ and $\textbf{G}_{1,k}$ with a non-negligible advantage.
That is, for any adversary $\mcA$, there exists a PPT algorithm $\mcB_3$ such
that
    $\big| \Adv_{\mcA}^{G'_{1,k}} - \Adv_{\mcA}^{G_{1,k}} \big| =
    \Adv_{\mcB_3}^{A3}(\lambda)$.
\end{lemma}

\begin{lemma} \label{lem:ahibe-prime-4}
If Assumption 4 holds, then no PPT algorithm can distinguish between
$\textbf{G}_2$ and $\textbf{G}_3$ with a non-negligible advantage. That is,
for any adversary $\mcA$, there exists a PPT algorithm $\mcB_4$ such that
    $\big| \Adv_{\mcA}^{G_2} - \Adv_{\mcA}^{G_3} \big| =
    \Adv_{\mcB_4}^{A4}(\lambda)$.
\end{lemma}

\begin{lemma} \label{lem:ahibe-prime-5}
If Assumption 5 holds, then no PPT algorithm can distinguish between
$\textbf{G}_3$ and $\textbf{G}_4$ with a non-negligible advantage. That is,
for any adversary $\mcA$, there exists a PPT algorithm $\mcB_5$ such that
    $\big| \Adv_{\mcA}^{G_3} - \Adv_{\mcA}^{G_4} \big| =
    \Adv_{\mcB_5}^{A5}(\lambda)$.
\end{lemma}

\noindent The security proof of Lemmas~\ref{lem:ahibe-prime-1},
\ref{lem:ahibe-prime-2}, \ref{lem:ahibe-prime-3}, \ref{lem:ahibe-prime-4},
and \ref{lem:ahibe-prime-5} are given in Section~\ref{sec:lem-proof}.

\subsection{Extensions}

\textbf{Relaxed Security Model.} The original security experiment of
anonymous HIBE requires that an adversary should select two hierarchical
identities $ID_0^*, ID_1^* \in \mathcal{I}^n$ with equal depth $n$
\cite{AbdallaBC+05}. One possible relaxation of the security experiment of
anonymous HIBE is to allow the adversary to select two hierarchical
identities $ID_0^* \in \mathcal{I}^{n_1}, ID_1^* \in \mathcal{I}^{n_2}$ with
different depths $n_1, n_2$. Our scheme is also fully secure in this relaxed
security experiment since the ciphertext size is constant. The two challenge
hierarchical identities with different depths only matter in the security
proof that distinguishes $\textbf{G}_3$ from $\textbf{G}_4$. In that proof,
we showed that the adversary cannot distinguish the challenge hierarchical
identity $ID_{\gamma}^*$ from a random value. Thus our scheme is secure in
this relaxed experiment since the ciphertext size does not reveal the depth
of the hierarchical identity.

\section{Performance Analysis}

In this section, we analyze the running time of our scheme, and then we
measure the performance of the scheme by implementing it.

\subsection{Runtime Analysis}

To analyze the efficiency of our scheme, we use the abstract cost of
expensive mathematical operations. In bilinear groups, the expensive
operations are exponentiation operations and pairing operations.
Additionally, the efficiency of exponentiations and pairings can be improved
by doing $m$-term exponentiations and $m$-term pairings respectively. The
abstract cost of these operations is defined as follows:

\begin{itemize}
\item $\textsf{MPairCost}(\G, \hat{\G}, m)$: $m$-term pairing
    $\prod_{i=1}^m e(g_i, \hat{h}_i)$ where $g_i \in \G, h_i \in
    \hat{\G}$
\item $\textsf{PairCost}(\G, \hat{\G})$: pairing $e(g,\hat{h})$ where $g
    \in \G, h \in \hat{\G}$
\item $\textsf{MExpCost}(\G, m)$: $m$-term exponentiation $\prod_{i=1}^m
    g_i^{a_i}$ where $g_i \in \G$
\item $\textsf{ExpCost}(\G)$: exponentiation $g^a$ where $g \in \G$
\end{itemize}

Let $l$ be the maximum number of hierarchical depth and $d$ be the depth of
$ID$. We define the abstract costs of the setup algorithm, the key generation
algorithm, the delegation algorithm, the encryption algorithm, and the
decryption algorithm as $\textsf{SetupCost}, \textsf{GenCost}, \textsf{DelCost},
\textsf{EncCost}, \textsf{DecCost}$ respectively. The abstract costs of these
algorithm are obtained as follows:
    \begin{align*}
    & \textsf{SetupCost}(l) \geq ~
        (2l+4) * \textsf{ExpCost}(\G) +
        2 * \textsf{ExpCost}(\hat{\G})
        + \textsf{PairCost}(\G,\hat{\G}), \\
    & \textsf{GenCost}(l,d) \geq ~
        (4(l-d)+4) * \textsf{ExpCost}(\hat{\G})
        + (2(l-d)+2) * \textsf{MExpCost}(\hat{\G}, 2) \\
    &   \qquad\qquad\qquad\quad
        + \frac{d}{m} * \textsf{MExpCost}(\hat{\G}, m),
        \displaybreak[0] \\
    & \textsf{DelCost}(l,d) \geq ~
        (6(l-d)+6) * \textsf{MExpCost}(\hat{\G}, 2)
        + 9 * \textsf{ExpCost}(\hat{\G}), \\
    & \textsf{EncCost}(d) \geq ~
        \frac{3d}{m} * \textsf{MExpCost}(\G, m) +
        6 * \textsf{ExpCost}(\G)
        + \textsf{ExpCost}(\G_T), \\
    & \textsf{DecCost} \geq ~
        2 * \textsf{MPairCost}(\G, \hat{\G}, 3).
    \end{align*}

In asymmetric bilinear groups, the bit size of $\hat{\G}$ and the bit size of
$\G_T$ increase proportionally to the embedding degree of asymmetric bilinear
groups. Thus the cost of exponentiation in $\hat{\G}$ is higher than the cost
of exponentiation in $\G$. In our scheme, the cost of the key generation
algorithm and the cost of the delegation algorithm are higher than the cost
of other algorithm since our scheme uses group elements in $\G$ for
ciphertexts and group elements in $\hat{\G}$ for private keys, and these
costs decrease proportionally to the depth of $ID$. The cost of the
encryption algorithm is small since it uses $m$-term exponentiations in $\G$,
and the cost of the decryption algorithm is constant.

\subsection{Implementation}

To show the efficiency of our scheme, we present the implementation of our
scheme and analyze the performance of it. We use the Pairing Based
Cryptography (PBC) library \cite{PBC} to implement our scheme, and we use a
notebook computer with an Intel Core i5 2.53 GHz CPU as a test machine. We
select a 175-bit Miyaji-Nakabayashi-Takano (MNT) curve with embedding degree
6. In the 175-bit MNT curve, the group size of $\G$ is about 175 bits, the
group size of $\hat{\G}$ is about 525 bits, and the group size of $\G_T$ is
about 1050 bits. The PBC library on the test machine can compute an
exponentiation of $\G$ in 1.6 ms, an exponentiation of $\hat{\G}$ in 20.3 ms,
an exponentiation of $\G_T$ in 4.7 ms, and a pairing in 15.6 ms.
Additionally, the PBC library can compute a three-term multi-exponentiation
of $\G$ in 2.1 ms, a two-term multi-exponentiation of $\hat{\G}$ in 27.3 ms,
a three-term multi-exponentiation of $\hat{\G}$ in 28.6 ms, and a three-term
multi-pairing in 31.2 ms. Therefore, we can obtain the cost of our scheme
using the 175-bit MNT curve on the test machine as follows:
    \begin{align*}
    & \textsf{GenCost}(l,d) \geq ~ 135.8 * (l-d) + 9.5 * d + 135.8 ~\mbox{ms}, \\
    & \textsf{DelCost}(l,d) \geq ~ 163.8 * (l-d) + 346.5 ~\mbox{ms}, \\
    & \textsf{EncCost}(d)   \geq ~ 2.1 * d + 14.3 ~\mbox{ms}, \\
    & \textsf{DecCost}      \geq ~ 62.4 ~\mbox{ms}.
    \end{align*}

\begin{figure*}[t]
\centering \small
\include{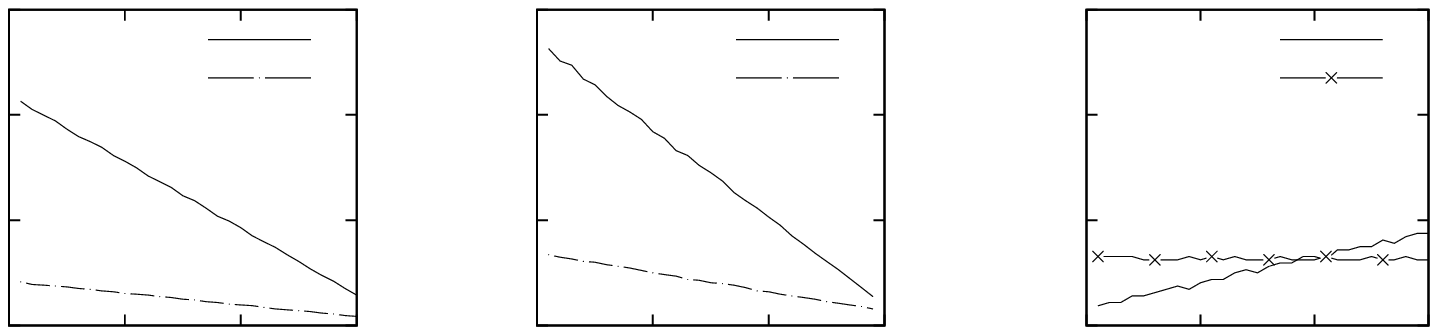}
\caption{Performance of our HIBE scheme} \label{fig:perf-hibe}
\end{figure*}

Let $l=30$. The performance results of each algorithms are described in
Figure~\ref{fig:perf-hibe}. The setup algorithm takes about 0.936 seconds to
generate the public parameters and the master key. The key generation
algorithm and the delegation algorithm for one depth take about 4.259 seconds
and 5.257 seconds respectively. One method to improve the performance of the
key generation algorithm is to preprocess the public parameters and the
master key. If the preprocessing method is used, then the cost of the key
generation algorithm is reduced to 1/5. This method also can be used in the
delegation algorithm.

\section{Proof of Lemmas} \label{sec:lem-proof}

In this section, we give the security proofs of Lemmas for our HIBE scheme.

\subsection{Proof of Lemma~\ref{lem:ahibe-prime-1}
(Indistinguishability of $\textbf{G}_0$ and $\textbf{G}_1$)}


In this proof, private keys are normal and the challenge ciphertext should be
normal or semi-functional depending on the $T$ value of the given assumption.
The main idea of this proof is that a simulator can only create normal
private keys since an element for semi-functional private keys is not given
in the assumption and the simulator embeds the $T$ element of the assumption
into the challenge ciphertext.

\vs \textbf{Simulator}. Suppose there exists an adversary $\mcA$ that
distinguishes between $\textbf{G}_0$ and $\textbf{G}_1$ with a non-negligible
advantage. A simulator $\mcB_1$ that breaks Assumption 1 using $\mcA$ is
given: a challenge tuple
    $D = ((p, \G, \hat{\G}, \G_T, e),
    k, k^a, k^b, k^{ab^2}, k^{b^2}, k^{b^3},
    k^c, k^{ac}, k^{bc}, k^{b^2 c}, k^{b^3 c}, \hat{k}, \hat{k}^b)$
and $T$ where $T = T_0 = k^{ab^2c}$ or $T = T_1 = k^{ab^2c + d}$. Then
$\mcB_1$ that interacts with $\mcA$ is described as follows:
$\mcB_1$ first chooses random exponents $\phi_2, B, \{ A_i \}_{i=1}^l, \alpha
\in \Z_p$ and random blinding values $y_g, y_h, \{ y_{u_i} \}_{i=1}^l, y_w
\in \Z_p$. It implicitly sets $\nu = a, \phi_1 = b, \tau = b + a \phi_2$ and
creates the public parameters as
    \begin{align*}
    &   g = k^{b^2} k^{y_g},~
        g^{\nu} = k^{ab^2} (k^a)^{y_g},~
        g^{-\tau} = (k^{b^3} (k^b)^{y_g}
            (k^{ab^2})^{\phi_2} (k^a)^{y_g \phi_2} )^{-1},~ \\
    &   h = (k^{b^2})^B k^{y_h},~
        h^{\nu} = (k^{ab^2})^B (k^a)^{y_h},~
        h^{-\tau} = ((k^{b^3})^B (k^b)^{y_h} (k^{ab^2})^{B \phi_2}
            (k^a)^{y_h \phi_2} )^{-1},~ \\
    &   \big\{
        u_i = (k^{b^2})^{A_i} k^{y_{u_i}},~
        u_i^{\nu} = (k^{ab^2})^{A_i} (k^a)^{y_{u_i}},~
        u_i^{-\tau} = ((k^{b^3})^{A_i} (k^b)^{y_{u_i}}
            (k^{ab^2})^{A_i \phi_2} (k^a)^{y_{u_i} \phi_2} )^{-1}
        \big\}_{i=1}^l,~ \\
    &   \hat{w}^{\phi_1} = (\hat{k}^b)^{y_w},~
        \hat{w}^{\phi_2} = \hat{k}^{y_w \phi_2},~
        \hat{w} = \hat{k}^{y_w},~
        \Omega = (e(k^{b^3}, \hat{k}^b) \cdot e(k^{b^2}, \hat{k})^{2 y_g}
            \cdot e(k, \hat{k})^{y_g^2})^{\alpha}.
    \end{align*}
It also implicitly sets
    $\hat{g} = \hat{k}^{b^2} \hat{k}^{y_g}, \hat{h} = \hat{k}^{b^2 B} \hat{k}^{y_h},
    \hat{u}_i = \hat{k}^{b^2 A_i} \hat{k}^{y_{u_i}}$
for the master key, but it cannot create these elements since $\hat{k}^{b^2}$
is not given. Additionally, it sets $f = k, \hat{f} = \hat{k}$ for the
semi-functional ciphertext and private key.
Let $\Delta(ID) = y_h + \sum_{i=1}^m y_{u_i} I_i$ and $\Gamma(ID) = B +
\sum_{i=1}^m A_i I_i$ where $ID = (I_1, \ldots, I_m)$. $\mcA$ adaptively
requests a private key for $ID = (I_1, \ldots, I_m)$. To response the private
key query, $\mcB_1$ first selects random exponents $r_1, c'_1, c'_2, \{
c'_{3,i} \}_{i=m+1}^l \in \Z_p$. It implicitly sets
    $c_1 = -b (\alpha + \Gamma(ID) r_1) / y_w + c'_1,~
    c_2 = -b r_1 / y_w + c'_2,~
    \{ c_{3,i} = -b A_i r_1 / y_w + c'_{3,i} \}_{i=m+1}^l$
and creates the decryption and delegation components of a private key as
    \begin{align*}
    &   K_{1,1} = \hat{k}^{y_g \alpha + \Delta(ID) r_1} (\hat{w}^{\phi_1})^{c'_1},~
        K_{1,2} = (K_{1,3})^{\phi_2},~
        K_{1,3} = (\hat{k}^b)^{-(\alpha + \Gamma(ID) r_1)} \hat{w}^{c'_1},~ \\
    &   K_{2,1} = \hat{k}^{y_g r_1} (\hat{w}^{\phi_1})^{c'_2},~
        K_{2,2} = (K_{2,3})^{\phi_2},~
        K_{2,3} = (\hat{k}^b)^{-r_1} \hat{w}^{c'_2},~ \\
    &   \big\{
        L_{3,i,1} = \hat{k}^{y_{u_i} r_1} (\hat{w}^{\phi_1})^{c'_{3,i}},~
        L_{3,i,2} = (L_{3,i,3})^{\phi_2},~
        L_{3,i,3} = (\hat{k}^b)^{-A_i r_1} \hat{w}^{c'_{3,i}}
        \big\}_{i=m+1}^l.
    \end{align*}
It also creates the randomization components of a private key similarly by
selecting random exponents $r_2, c'_4, c'_5, \{ c'_{6,i} \}_{i=n+1}^l \in
\Z_p$ except that $R_{1,1}$ does not have $\hat{g}^{\alpha}$. We omit the
detailed description of these.
In the challenge step, $\mcA$ submits two challenge hierarchical identities
$ID_0^* = (I_{0,1}^*, \ldots, I_{0,n}^*), ID_1^* = (I_{1,1}^*, \ldots,
I_{1,n}^*)$ and two messages $M_0^*, M_1^*$. $\mcB_1$ flips a random coin
$\gamma \in \{0,1\}$ internally. It implicitly sets $t = c$ and creates a
challenge ciphertext as
    \begin{align*}
    &   C       = (e(k^{b^3 c}, \hat{k}^b) \cdot e(k^{b^2 c}, \hat{k})^{2 y_g}
                  \cdot e(k^c, \hat{k})^{y_g^2})^{\alpha} \cdot M_{\gamma}^*,~ \\
    &   C_{1,1} = k^{b^2 c} (k^c)^{y_g},~
        C_{1,2} = T (k^{ac})^{y_g},~
        C_{1,3} = ((k^{b^3 c}) (k^{bc})^{y_g} (T)^{\phi_2} (k^{ac})^{y_g \phi_2})^{-1}, \\
    &   C_{2,1} = (k^{b^2 c})^{\Gamma(ID_{\gamma}^*)} (k^c)^{\Delta(ID_{\gamma}^*)},~
        C_{2,2} = (T)^{\Gamma(ID_{\gamma}^*)}
                  (k^{ac})^{\Delta(ID_{\gamma}^*)},~ \\
    &   C_{2,3} = \big(
                  (k^{b^3 c})^{\Gamma(ID_{\gamma}^*)}
                  (k^{bc})^{\Delta(ID_{\gamma}^*)}
                  (T)^{\phi_2 \Gamma(ID_{\gamma}^*)}
                  (k^{ac})^{\phi_2 \Delta(ID_{\gamma}^*)} \big)^{-1}.
    \end{align*}
Finally, $\mcA$ outputs a guess $\gamma'$. If $\gamma = \gamma'$, $\mcB_1$
outputs 0. Otherwise, it outputs 1.

\vs \textbf{Analysis}. We first show that the distribution of the simulation
using $D, T = T_0 = k^{ab^2c}$ is the same as $\textbf{G}_0$. The public
parameters are correctly distributed since the random blinding values $y_g,
y_h, \{y_{u_i}\}, y_w$ are used. The private key is correctly distributed as
    \begin{align*}
    K_{1,1}
        &= \hat{g}^{\alpha} (\hat{h} \prod_{i=1}^m \hat{u}_i^{I_i})^{r_1}
            (\hat{w}^{\phi_1})^{c_1}
         = (\hat{k}^{b^2 + y_g})^{\alpha} (\hat{k}^{b^2 B + y_h}
            \prod_{i=1}^m \hat{k}^{(b^2 A_i + y_{u_i})I_i})^{r_1}
            (\hat{k}^{b y_w})^{-b(\alpha + \Gamma(ID) r_1) / y_w + c'_1} \\
        &= \hat{k}^{y_g \alpha + \Delta(ID) r_1} (\hat{w}^{\phi_1})^{c'_1},~
            \displaybreak[0] \\
    K_{2,1}
        &= \hat{g}^{r_1} (\hat{w}^{\phi_1})^{c_2}
         = (\hat{k}^{b^2 + y_g})^{r_1} (\hat{k}^{b y_w})^{-b r_1 / y_w + c'_2}
         = \hat{k}^{y_g r_1} (\hat{w}^{\phi_1})^{c'_2},~ \\
    L_{3,i,1}
        &= \hat{u}_i^{r_1} (\hat{w}^{\phi_1})^{c_{3,i}}
         = (\hat{k}^{b^2 A_i + y_{u_i}})^{r_1}
           (\hat{k}^{b y_w})^{-b A_i r_1 / y_w + c'_{3,i}}
         = \hat{k}^{y_{u_i} r_1} (\hat{w}^{\phi_1})^{c'_{3,i}}.
    \end{align*}
Note that it can create a normal private key since $c_1, c_2, \{ c_{3,i} \},
c_4, c_5, \{ c_{6,i} \}$ enable the cancellation of $\hat{k}^{b^2}$, but it
cannot create a semi-functional private key since $\hat{k}^a$ is not given.
The challenge ciphertext is correctly distributed as
    \begin{align*}
    C_{1,1} &= g^t = (k^{b^2 + y_g})^c = k^{b^2 c} (k^c)^{y_g},~
    C_{1,2}  = (g^{\nu})^t = k^{(b^2 + y_g)ac} = T_0 (k^{ac})^{y_g},~ \\
    C_{1,3} &= (g^{-\tau})^t = (k^{(b^2 + y_g) (b + a \phi_2) c})^{-1}
             = ((k^{b^3 c}) (k^{bc})^{y_g} (T_0)^{\phi_2} (k^{ac})^{y_g \phi_2})^{-1},
             \displaybreak[0] \\
    C_{2,1} &= (h \prod_{i=1}^n u_i^{I_{\gamma,i}^*})^t
             = (k^{b^2 B + y_h} \prod_{i=1}^n k^{(b^2 A_i + y_{u_i}) I_{\gamma,i}^*})^c
             = (k^{b^2 c})^{\Gamma(ID_{\gamma}^*)} (k^c)^{\Delta(ID_{\gamma}^*)},~ \\
    C_{2,2} &= (h^{\nu} \prod_{i=1}^n (u_i^{\nu})^{I_{\gamma,i}^*})^t
             = (k^{(b^2 B + y_h)a} \prod_{i=1}^n k^{(b^2 A_i + y_{u_i})a I_{\gamma,i}^*})^c
             = (T_0)^{\Gamma(ID_{\gamma}^*)} (k^{ac})^{\Delta(ID_{\gamma}^*)},~
             \displaybreak[0] \\
    C_{2,3} &= (h^{-\tau} \prod_{i=1}^n (u_i^{-\tau})^{I_{\gamma,i}^*})^t
             = ((k^{(b^2 B + y_h)(b + a \phi_2)} \prod_{i=1}^n
               k^{(b^2 A_i + y_{u_i})(b + a \phi_2) I_{\gamma,i}^*})^c)^{-1} \\
            &= ((k^{b^3 c})^{\Gamma(ID_{\gamma}^*)} (k^{bc})^{\Delta(ID_{\gamma}^*)}
               (T_0)^{\phi_2 \Gamma(ID_{\gamma}^*)}
               (k^{ac})^{\phi_2 \Delta(ID_{\gamma}^*)})^{-1}.
    \end{align*}
We next show that the distribution of the simulation using $D, T = T_1 =
k^{ab^2c + d}$ is the same as $\textbf{G}_1$. We only consider the
distribution of the challenge ciphertext since $T$ is only used in the
challenge ciphertext. The only difference between $T_0$ and $T_1$ is that
$T_1$ additionally has $k^d$. Thus $C_{1,2}, C_{1,3}, C_{2,2}, C_{2,3}$
components that have $T$ in the simulation additionally have
    $k^d, (k^d)^{-\phi_2},
    (k^d)^{\Gamma(ID_{\gamma}^*)},
    (k^d)^{-\phi_2 \Gamma(ID_{\gamma}^*)}$
respectively. If we implicitly set $s_c = d, z_c = \Gamma(ID_{\gamma}^*)$,
then the challenge ciphertext is semi-functional. The distribution of this
semi-functional challenge ciphertext is the same as $\textbf{G}_1$ since $B,
\{ A_i \}$ for $z_c$ are information theoretically hidden to $\mcA$.
We obtain $\Pr [\mcB_1(D,T_0) = 0] - 1/2 = \Adv_{\mcA}^{G_0}$ and $\Pr
[\mcB_1(D,T_1) = 0] - 1/2 = \Adv_{\mcA}^{G_1}$ from the above analysis. Thus,
we can easily derive the advantage of $\mcB_1$ as
    \begin{align*}
    \Adv_{\mcB_1}^{A1}(\lambda)
     =  \big| \Pr[\mcB_1(D, T_0) = 0] - \Pr[\mcB_1(D, T_1) = 0] \big|
     =  \big| \Adv_{\mcA}^{G_0} - \Adv_{\mcA}^{G_1} \big|.
    \end{align*}
This completes our proof.

\subsection{Proof of Lemma~\ref{lem:ahibe-prime-2}
(Indistinguishability of $\textbf{G}_{1,k-1}$ and $\textbf{G}'_{1,k}$)}


In this proof, the challenge ciphertext is semi-functional and the $k$-th
private key should be normal or semi-functional type-1 depending on the $T$
value of the given assumption. However, the paradox of dual system encryption
occurs in this proof since a simulator can create a semi-functional
ciphertext to check the type of the $k$-th private key by decrypting the
semi-functional ciphertext using the $k$-the private key. The main idea to
solve this paradox is to use a nominally semi-functional type-1 private key.
If the $k$-th private key is nominally semi-functional type-1, then $z_{k,1}$
of the nominally semi-functional private key is the same as the $z_c$ of a
semi-functional challenge ciphertext. Thus the simulator cannot distinguish
the type of $k$-th private key since the decryption of the semi-functional
ciphertext using the $k$-th private key always succeeds.

Before proving this lemma, we introduce Assumption 2-A as follows: Let $(p,
\G, \hat{\G}, \G_T, e)$ be a description of the asymmetric bilinear group of
prime order $p$. Let $g, \hat{g}$ be generators of $\G, \hat{\G}$
respectively. Assumption 2-A is that if the challenge values
    $D = ( (p, \G, \hat{\G}, \G_T, e),
    k, k^a, k^{a^2}, k^{bx}, k^{abx}, k^{a^2x}, \hat{k}, \hat{k}^a,
    \hat{k}^b, \hat{k}^{y_1}, \hat{k}^{y_2})$
and $T = (D_1, D_2)$ are given, no PPT algorithm can distinguish $T =
(\hat{k}^{b y_1}, \hat{k}^{b y_2})$ from $T = (\hat{k}^{d_1}, \hat{k}^{d_2})$
with more than a negligible advantage. It is easy to show that if there
exists an adversary that breaks Assumption 2-A, then an algorithm can break
Assumption 2 with the same probability by setting
    $\hat{k}^{y_1} = (\hat{k}^{b})^{r_1} \hat{k}^{s_1},
     \hat{k}^{y_2} = (\hat{k}^{b})^{r_2} \hat{k}^{s_2},
     D_1 = (T)^{r_1} (\hat{k}^c)^{s_1},
     D_2 = (T)^{r_1} (\hat{k}^c)^{s_1}$
where $\hat{k}^b, \hat{k}^c, T$ are given in Assumption 2 and $r_1, r_2, s_1,
s_2$ are random exponents in $\Z_p$. The simulated values are correctly
distributed since there exists one-to-one correspondence between $\{r_1, s_1,
r_2, s_2\}$ and $\{y_1, y_2, d_1, d_2\}$.

\vs \textbf{Simulator}. Suppose there exists an adversary $\mcA$ that
distinguishes between $\textbf{G}_{1,k-1}$ and $\textbf{G}'_{1,k}$ with a
non-negligible advantage. A simulator $\mcB_2$ that breaks Assumption 2-A
using $\mcA$ is given: a challenge tuple
    $D = ((p, \G, \hat{\G}, \G_T, e),
    k, k^a, k^{a^2}, k^{bx}, k^{abx}, k^{a^2x},
    \hat{k}, \hat{k}^a, \hat{k}^b, \hat{k}^{y_1}, \hat{k}^{y_2})$
    and $T = (D_1, D_2)$ where
    $T = T_0 = (D_1^0, D_2^0) = (\hat{k}^{b y_1}, \hat{k}^{b y_2})$ or
    $T = T_1 = (D_1^1, D_2^1) = (\hat{k}^{b y_1 + d_1}, \hat{k}^{b y_2 + d_2})$.
Then $\mcB_2$ that interacts with $\mcA$ is described as follows:
$\mcB_2$ first chooses random exponents $\nu, y_{\tau}, B, \{ A_i \}_{i=1}^l,
\alpha \in \Z_p$ and random blinding values $y_h, \{ y_{u_i} \}_{i=1}^l, y_w
\in \Z_p$. It implicitly sets
    $\phi_1 = -\nu b + (a + y_{\tau}), \phi_2 = b, \tau = a + y_{\tau}$
and creates the public parameters as
    \begin{align*}
    &   g = k^a,~
        g^{\nu} = (k^a)^{\nu},~
        g^{-\tau} = (k^{a^2} (k^a)^{y_{\tau}})^{-1},~ \\
    &   h = (k^a)^B k^{y_h},~
        h^{\nu} = (k^a)^{B \nu} k^{y_h \nu},~
        h^{-\tau} = ((k^{a^2})^B (k^a)^{y_h + B y_{\tau}} k^{y_h y_{\tau}})^{-1},~ \db \\
    &   \big\{
        u_i = (k^a)^{A_i} k^{y_{u_i}},~
        u_i^{\nu} = (k^a)^{A_i \nu} k^{y_{u_i} \nu},~
        u_i^{-\tau} = ((k^{a^2})^{A_i} (k^a)^{y_{u_i} + A_i y_{\tau}}
            k^{y_{u_i} y_{\tau}})^{-1}
        \big\}_{i=1}^l,~ \\
    &   \hat{w}^{\phi_1} = ((\hat{k}^b)^{-\nu} \hat{k}^a \hat{k}^{y_{\tau}})^{y_w},~
        \hat{w}^{\phi_2} = (\hat{k}^b)^{y_w},~
        \hat{w} = \hat{k}^{y_w},~
        \Omega = e(k^a, \hat{k}^a)^{\alpha}.
    \end{align*}
It also sets
    $\hat{g} = \hat{k}^a, \hat{g}^{\alpha} = (\hat{k}^a)^{\alpha},
    \hat{h} = (\hat{k}^a)^B \hat{k}^{y_h},
    \{ \hat{u}_i = (\hat{k}^a)^{A_i} \hat{k}^{y_{u_i}} \}_{i=1}^l$
for the master key. Additionally, it sets $f = k, \hat{f} = \hat{k}$ for the
semi-functional ciphertext and private key.
Let $\Delta(ID) = y_h + \sum_{i=1}^m y_{u_i} I_i$ and $\Gamma(ID) = B +
\sum_{i=1}^m A_i I_i$ where $ID = (I_1, \ldots, I_m)$. $\mcA$ adaptively
requests a private key for $ID = (I_1, \ldots, I_m)$. If this is a $j$-th
private key query, then $\mcB_2$ handles this query as follows:
\begin{itemize}
\item \textbf{Case $j < k$} : It creates a semi-functional private key by
    calling \textbf{KeyGenSF-2} since it knows the master key and the
    tuple $(\hat{f}^{-\nu}, \hat{f}, 1)$ for the semi-functional private
    key.

\item \textbf{Case $j = k$} : It first selects random exponents $r'_1,
    c'_1, c'_2, \{ c'_{3,i} \}_{i=m+1}^l \in \Z_p$. It implicitly sets
    $r_1 = -y_1 + r'_1,~ c_1 = y_1 \Gamma(ID) / y_w + c'_1,~ c_2 = y_1 /
    y_w + c'_2,~ \{ c_{3,i} = y_1 A_i / y_w + c'_{3,i} \}_{i=m+1}^l$ and
    creates the decryption and delegation components of a private key as
    \begin{align*}
    &   K_{1,1} = \hat{g}^{\alpha} (\hat{k}^{y_1})^{-\Delta(ID)}
            (\hat{h} \prod_{i=1}^m \hat{u}_i^{I_i})^{r'_1} (D_1)^{-\nu \Gamma(ID)}
            (\hat{k}^{y_1})^{y_{\tau} \Gamma(ID)} (\hat{w}^{\phi_1})^{c'_1}, \\
    &   K_{1,2} = (D_1)^{\Gamma(ID)} (\hat{w}^{\phi_2})^{c'_1},~
        K_{1,3} = (\hat{k}^{y_1})^{\Gamma(ID)} \hat{w}^{c'_1},~
        \displaybreak[0] \\
    &   K_{2,1} = \hat{g}^{r'_1} (D_1)^{-\nu} (\hat{k}^{y_1})^{y_{\tau}}
            (\hat{w}^{\phi_1})^{c'_2},
        K_{2,2} = D_1 (\hat{w}^{\phi_2})^{c'_2},~
        K_{2,3} = \hat{k}^{y_1} \hat{w}^{c'_2},~ \\
    &   \big\{
        L_{3,i,1} = (\hat{k}^{y_1})^{-y_{u_i}} \hat{u}_i^{r'_1}
            (D_1)^{-\nu A_i} (\hat{k}^{y_1})^{y_{\tau} A_i}
            (\hat{w}^{\phi_1})^{c'_{3,i}},~
        L_{3,i,2} = (D_1)^{A_i} (\hat{w}^{\phi_2})^{c'_{3,i}},~ \\
    &~~ L_{3,i,3} = (\hat{k}^{y_1})^{A_i} \hat{w}^{c'_{3,i}}
        \big\}_{i=m+1}^l.
    \end{align*}
    It also creates the randomization components of a private key
    similarly by selecting random exponents $r'_2, c'_4, c'_5, \{
    c'_{6,i} \}_{i=m+1}^l \in \Z_p$ except that it uses $\hat{k}^{y_2},
    D_2$ instead of $\hat{k}^{y_1}, D_1$. We omit the detailed
    description of these.

\item \textbf{Case $j > k$} : It creates a normal private key by calling
    \textbf{KeyGen} since it knows the master key.
\end{itemize}
In the challenge step, $\mcA$ submits two challenge hierarchical identities
$ID_0^* = (I_{0,1}^*, \ldots, I_{0,n}^*), ID_1^* = (I_{1,1}^*, \ldots,
I_{1,n}^*)$ and two messages $M_0^*, M_1^*$. $\mcB_2$ flips a random coin
$\gamma \in \{0,1\}$ internally and chooses a random exponent $t' \in \Z_p$.
It implicitly sets
    $t = bx + t',~ s_c = -a^2 x,~ z_c = \Gamma(ID_{\gamma}^*)$
and creates a semi-functional ciphertext as
    \begin{align*}
    C ~~~~  &= e(k^{abx}, \hat{k}^a)^{\alpha} \cdot e(k^a, \hat{k}^a)^{\alpha t'}
               \cdot M_{\gamma}^*,~ \\
    C_{1,1} &= k^{abx} g^{t'},~
    C_{1,2}  = (k^{abx})^{\nu} (g^{\nu})^{t'} (k^{a^2 x})^{-1},~
    C_{1,3}  = (k^{abx})^{-y_{\tau}} (g^{-\tau})^{t'}, \\
    C_{2,1} &= (k^{abx})^{\Gamma(ID_{\gamma}^*)}
               (k^{bx})^{\Delta(ID_{\gamma}^*)}
               (h \prod_{i=1}^n u_i^{I_{\gamma,i}^*})^{t'},~ \\
    C_{2,2} &= (k^{abx})^{\Gamma(ID_{\gamma}^*) \nu}
               (k^{bx})^{\Delta(ID_{\gamma}^*) \nu}
               (h^{\nu} \prod_{i=1}^n (u_i^{\nu})^{I_{\gamma,i}^*})^{t'}
               (k^{a^2 x})^{-\Gamma(ID_{\gamma}^*)},~ \\
    C_{2,3} &= (k^{abx})^{-\Gamma(ID_{\gamma}^*) y_{\tau}}
               (k^{abx})^{-\Delta(ID_{\gamma}^*)}
               (k^{bx})^{-\Delta(ID_{\gamma}^*) y_{\tau}}
               (h^{-\tau} \prod_{i=1}^n (u_i^{-\tau})^{I_{\gamma,i}^*})^{t'}.
    \end{align*}
Finally, $\mcA$ outputs a guess $\gamma'$. If $\gamma = \gamma'$, $\mcB_2$
outputs 0. Otherwise, it outputs 1.

\vs \textbf{Analysis}. We first show that the distribution of the simulation
using $D, T_0 = (D_1^0, D_2^0) = (\hat{k}^{b y_1}, \hat{k}^{b y_2})$ is the
same as $\textbf{G}_{1,k-1}$. The public parameters are correctly distributed
since the random blinding values $y_h, \{ y_{u_i} \}, y_w$ are used. The
$k$-th private key is correctly distributed as
    \begin{align*}
    K_{1,1} &= \hat{g}^{\alpha} (\hat{h} \prod_{i=1}^m \hat{u}_i^{I_i})^{r_1}
               (\hat{w}^{\phi_1})^{c_1}
             = \hat{g}^{\alpha} (\hat{k}^{a B + y_h} \prod_{i=1}^m
               \hat{k}^{(a A_i + y_{u_i}) I_i})^{-y_1 + r'_1}
               (\hat{k}^{y_w (-\nu b + a + y_{\tau})})^{y_1 \Gamma(ID) / y_w + c'_1} \\
            &= \hat{g}^{\alpha} (\hat{k}^{y_1})^{-\Delta(ID)}
               (\hat{h} \prod_{i=1}^m \hat{u}_i^{I_i})^{r'_1}
               (D_1^0)^{-\nu \Gamma(ID)}
               (\hat{k}^{y_1})^{y_{\tau} \Gamma(ID)} (\hat{w}^{\phi_1})^{c'_1},
               \displaybreak[0] \\
    K_{2,1} &= \hat{g}^{r_1} (\hat{w}^{\phi_1})^{c_2}
             = (\hat{k}^a)^{-y_1 + r'_1} (\hat{k}^{y_w (-\nu b + a + y_{\tau})}
                )^{y_1 / y_w + c'_2}
             = \hat{g}^{r'_1} (D_1^0)^{-\nu} (\hat{k}^{y_1})^{y_{\tau}}
               (\hat{w}^{\phi_1})^{c'_2}, \\
    L_{3,i,1} &= \hat{u}_i^{r_1} (\hat{w}^{\phi_1})^{c_{3,i}}
             = (\hat{k}^{a A_i + y_{u_i}})^{-y_1 + r'_1}
               (\hat{k}^{y_w (-\nu b + a + y_{\tau})})^{y_1 A_i / y_w + c'_{3,i}} \\
            &= (\hat{k}^{y_1})^{-y_{u_i}} \hat{u}_i^{r'_1}
               (D_1^0)^{-\nu A_i} (\hat{k}^{y_1})^{y_{\tau} A_i}
               (\hat{w}^{\phi_1})^{c'_{3,i}}.
    \end{align*}
The semi-functional challenge ciphertext is correctly distributed as
    \begin{align*}
    C_{1,1} &= g^t = (k^a)^{bx + t'} = k^{abx} g^{t'}, \\
    C_{1,2} &= (g^{\nu})^t f^{s_c} = (k^{a \nu})^{bx + t'} k^{-a^2 x}
             = (k^{abx})^{\nu} (g^{\nu})^{t'} (k^{a^2 x})^{-1}, \\
    C_{1,3} &= (g^{-\tau})^t (f^{-\phi_2})^{s_c}
             = (k^{a(-a - y_{\tau})})^{bx + t'} k^{-b (-a^2 x)}
             = (k^{abx})^{-y_{\tau}} (g^{-\tau})^{t'}, \displaybreak[0] \\
    C_{2,1} &= (h \prod_{i=1}^n u_i^{I_{\gamma,i}^*})^t
             = (k^{aB + y_h} \prod_{i=1}^n (k^{a A_i + y_{u_i}})^{I_{\gamma,i}^*})^{bx + t'}
             = (k^{abx})^{\Gamma(ID_{\gamma}^*)} (k^{bx})^{\Delta(ID_{\gamma}^*)}
               (h \prod_{i=1}^n u_i^{I_{\gamma,i}^*})^{t'},
               \displaybreak[0] \\
    C_{2,2} &= (h^{\nu} \prod_{i=1}^n (u_i^{\nu})^{I_{\gamma,i}^*})^t (f^{s_c})^{z_c}
             = (k^{(a B + y_h) \nu} \prod_{i=1}^n (k^{(a A_i + y_{u_i})
                    \nu})^{I_{\gamma,i}^*})^{bx + t'}
               k^{-a^2 x \Gamma(ID_{\gamma}^*)} \\
            &= (k^{abx})^{\Gamma(ID_{\gamma}^*) \nu} (k^{bx})^{\Delta(ID_{\gamma}^*) \nu}
               (h^{\nu} \prod_{i=1}^n (u_i^{\nu})^{I_{\gamma,i}^*})^{t'}
               (k^{a^2 x})^{-\Gamma(ID_{\gamma}^*)},
               \displaybreak[0] \\
    C_{2,3} &= (h^{-\tau} \prod_{i=1}^n (u_i^{-\tau})^{I_{\gamma,i}^*})^t
                    (f^{-\phi_2})^{s_c z_c} \\
            &= (k^{(a B + y_h) (-a - y_{\tau})} \prod_{i=1}^n
               (k^{(a A_i + y_{u_i}) (-a - y_{\tau})})^{I_{\gamma,i}^*})^{bx + t'}
               k^{-b (-a^2 x) \Gamma(ID_{\gamma}^*)} \\
            &= (k^{abx})^{-\Gamma(ID_{\gamma}^*) y_{\tau}}
               (k^{abx})^{-\Delta(ID_{\gamma}^*)}
               (k^{bx})^{-\Delta(ID_{\gamma}^*) y_{\tau}}
               (h^{-\tau} \prod_{i=1}^n (u_i^{-\tau})^{I_{\gamma,i}^*})^{t'}.
    \end{align*}
Note that it can create the semi-functional ciphertext with only fixed $z_c =
\Gamma(ID_{\gamma}^*)$ since $s_c, z_c$ enable the cancellation of $k^{a^2
bx}$. Even though the simulator uses the fixed $z_c$, the distribution of
$z_c$ is correct since $B, \{ A_i \}$ for $z_c$ are information theoretically
hidden to $\mcA$.
We next show that the distribution of the simulation using $D, T_1 = (D_1^1,
D_2^1) = (\hat{k}^{b y_1 + d_1}, \hat{k}^{b y_2 + d_2})$ is the same as
$\textbf{G}'_{1,k}$ except the $k$-th private key is nominally
semi-functional. We only consider the distribution of the $k$-th private key
since $T = (D_1, D_2)$ is only used in the $k$-th private key. The only
difference between $T_0 = (D_1^0, D_2^0)$ and $T_1 = (D_1^1, D_2^1)$ is that
$T_1 = (D_1^1, D_2^1)$ additionally has $(\hat{k}^{d_1}, \hat{k}^{d_2})$. The
decryption and delegation components $K_{1,1}, K_{1,2}, K_{2,1}, K_{2,2}, \{
L_{3,i,1}, L_{3,i,2} \}$ that have $D_1$ in the simulation additionally have
    $(\hat{k}^{d_1})^{-\nu \Gamma(ID)}, (\hat{k}^{d_1})^{\Gamma(ID)}$,
    $(\hat{k}^{d_1})^{-\nu}, \hat{k}^{d_1},
    \{ (\hat{k}^{d_1})^{-\nu A_i}, (\hat{k}^{d_1})^{A_i} \}$
respectively. The randomization components $R_{1,1}, R_{1,2}, R_{2,1},
R_{2,2}, \{ R_{3,i,1}, R_{3,i,2} \}$ that have $D_2$ in the simulation also
have the additional values except that $\hat{k}^{d_2}$ is used instead of
$\hat{k}^{d_1}$. If we implicitly set
    $s_{k,1} = d_1, z_{k,1} = \Gamma(ID),
     \{ z_{k,2,i} = A_i \}_{i=m+1}^l, s_{k,2} = d_2$,
then the distribution of the $k$-th private key is the same as
$\textbf{G}'_{1,k}$ except that the $k$-the private key is nominally
semi-functional type-1.

Finally, we show that the adversary cannot distinguish the nominally
semi-functional type-1 private key from the semi-functional type-1 private
key. The main idea of this proof is that the adversary cannot request a
private key for $ID$ that is a prefix of a challenge identity $ID^*$ in the
security model. Suppose there exists an unbounded adversary, then the
adversary can gather the values $z_{k,1} = \Gamma(ID) = B + \sum_{i=1}^m A_i
I_i, \{ z_{k,2,i} = A_i \}_{i=m+1}^l$ from the $k$-the private key query for
$ID = (I_1, \ldots, I_m)$ and $z_c = \Gamma(ID_{\gamma}^*) = B + \sum_{i=1}^n
A_i I_{\gamma,i}^*$ from the challenge ciphertext for $ID_{\gamma}^* =
(I_{\gamma,1}^*, \ldots, I_{\gamma,n}^*)$. In case of $n \geq m$, the values
that are revealed to the adversary are described as
    \begin{align*}
    \begin{pmatrix}
        1 & I_{\gamma,1}^* ~\cdots~ I_{\gamma,m}^* & I_{\gamma,m+1}^* & \cdots & 0 \\
        1 & I_1           ~~\cdots~~ I_m           & 0 & \cdots & 0 \\
        0 & 0             ~~\cdots~~ 0             & 1 & \cdots & 0 \\
        \vdots & \vdots   ~~\ddots~~ \vdots        & \vdots & \ddots & \vdots \\
        0 & 0             ~~\cdots~~ 0             & 0 & \cdots & 1 \\
    \end{pmatrix}
    \begin{pmatrix}
        B \\ A_1 \\ \vdots \\ A_m \\
        A_{m+1} \\ \vdots \\ A_l \\
    \end{pmatrix}
    =
    \begin{pmatrix}
        z_c \\ z_{k,1} \\
        z_{k,2,m+1} \\ \vdots \\ z_{k,2,l} \\
    \end{pmatrix}.
    \end{align*}
It is easy to show that the row rank of the above $(l-m+2) \times (l+1)$
matrix is $l-m+2$ since there exists an index $j$ such that $I_j \neq
I_{\gamma,j}^*$. It means that the above matrix is non-singular. In case of
$n < m$, the revealed values to the adversary also can be described as a
similar matrix equation as the above one. The row rank of this $(l-m+2)
\times (l+1)$ matrix is $l-m+2$ since $I_m \neq 0$. Therefore these values
look random to the unbounded adversary since the matrixes for two cases are
non-singular and $B, A_1, \ldots, A_l$ are chosen randomly.
We obtain $\Pr [\mcB_2(D,T_0) = 0] - 1/2 = \Adv_{\mcA}^{G_{1,k-1}}$ and $\Pr
[\mcB_2(D,T_1) = 0] - 1/2 = \Adv_{\mcA}^{G'_{1,k}}$ from the above analysis.
Thus, we can easily derive the advantage of $\mcB_2$ as
    \begin{align*}
    \Adv_{\mcB_2}^{A2}(\lambda)
     =  \big| \Pr[\mcB_2(D, T_0) = 0] - \Pr[\mcB_2(D, T_1) = 0] \big|
     =  \big| \Adv_{\mcA}^{G_{1,k-1}} - \Adv_{\mcA}^{G'_{1,k}} \big|.
    \end{align*}
This completes our proof.

\subsection{Proof of Lemma~\ref{lem:ahibe-prime-3}
(Indistinguishability of $\textbf{G}'_{1,k}$ and $\textbf{G}_{1,k}$)}


In this proof, the challenge ciphertext is semi-functional and the $k$-th
private key should be semi-functional type-1 or semi-functional type-2
depending on the $T$ value of the given assumption. The main idea of this
proof is to show that the semi-functional type-1 and semi-functional type-2
private keys are computationally indistinguishable using the given
assumption.

Before proving this lemma, we introduce Assumption 3-A as follows: Let $(p,
\G, \hat{\G}, \G_T, e)$ be a description of the asymmetric bilinear group of
prime order $p$. Let $g, \hat{g}$ be generators of $\G, \hat{\G}$
respectively. Assumption 3-A is that if the challenge values
    $D = ( (p, \G, \hat{\G}, \G_T, e),
    k, \hat{k}, \hat{k}^{x_1}, \hat{k}^{x_{2,1}}, \ldots, \hat{k}^{x_{2,l}},
    \hat{k}^y) \mbox{ and } T = (D_1, D_{2,1}, \ldots, D_{2,l})$
are given, no PPT algorithm can distinguish
    $T = T_0 = (\hat{k}^{x_1 y}, \hat{k}^{x_{2,1} y}, \ldots, \hat{k}^{x_{2,l} y})$
    from
    $T = T_1 = (\hat{k}^{d_1}, \hat{k}^{d_{2,1}}, \ldots, \hat{k}^{d_{2,l}})$
with more than a negligible advantage. It is easy to show that if there
exists an adversary that breaks Assumption 3-A, then an algorithm can break
Assumption 3 with the same probability by setting
    $\hat{k}^{x_1} = (\hat{k}^a)^{r_1} \hat{k}^{s_1},
     \{ \hat{k}^{x_{2,i}} = (\hat{k}^a)^{r_{2,i}} \hat{k}^{s_{2,i}} \}_{i=1}^{l},
     \hat{k}^{y} = \hat{k}^b,
     D_1 = (T)^{r_1} (\hat{k}^b)^{s_1},
     \{ D_{2,i} = (T)^{r_{2,i}} (\hat{k}^b)^{s_{2,i}} \}_{i=1}^{l}$
where $\hat{k}^a, \hat{k}^b, T$ are given in Assumption 3 and $r_1, s_1, \{
r_{2,i}, s_{2,i} \}_{i=1}^{l}$ are random exponents in $\Z_p$. The simulated
values are correctly distributed since there exists one-to-one correspondence
between $\{r_1, s_1, \{ r_{2,i} \}, \{ s_{2,i} \}\}$ and $\{x_1, \{ x_{2,i}
\}, d_1, \{ d_{2,i} \}\}$.

\vs \textbf{Simulator}. Suppose there exists an adversary $\mcA$ that
distinguishes between $\textbf{G}'_{1,k}$ and $\textbf{G}_{1,k}$ with a
non-negligible advantage. A simulator $\mcB_3$ that breaks Assumption 3-A
using $\mcA$ is given: a challenge tuple
    $D = ((p, \G, \hat{\G}, \G_T, e),
    k, \hat{k}, \hat{k}^{x_1}, \hat{k}^{x_{2,1}}, \ldots, \hat{k}^{x_{2,l}},
    \hat{k}^y)$
and $T = (D_1, \ldots, D_{2,l})$ where
    $T = T_0 = (D_1^0, \ldots, D_{2,l}^0)
       = (\hat{k}^{x_1 y}, \hat{k}^{x_{2,1} y}, \ldots, \hat{k}^{x_{2,l} y})$ or
    $T = T_1 = (D_1^1, \ldots, D_{2,l}^1)
       = (\hat{k}^{d_1}, \hat{k}^{d_{2,1}}, \ldots, \hat{k}^{d_{2,l}})$.
Then $\mcB_3$ that interacts with $\mcA$ is described as follows:
$\mcB_3$ first chooses random exponents $\nu, \phi_1, \phi_2, \alpha \in
\Z_p$ and random blinding values $y_g, y_h, \{ y_{u_i} \}_{i=1}^l, y_w \in
\Z_p$. It implicitly sets $\tau = \phi_1 + \nu \phi_2$ and sets
    $g = k^{y_g}, \hat{g} = \hat{k}^{y_g}, h = k^{y_h}, \hat{h} = \hat{k}^{y_h},
     \{ u_i = k^{y_{u_i}}, \hat{u}_i = \hat{k}^{y_{u_i}} \}_{i=1}^l,
     \hat{w} = \hat{k}^{y_w}$.
It creates the public parameters as
    \begin{align*}
    PP = \big(
        g, g^{\nu}, g^{-\tau},~ h, h^{\nu}, h^{-\tau},~
        \{ u_i, u_i^{\nu}, u_i^{-\tau} \}_{i=1}^l,~
        \hat{w}^{\phi_1}, \hat{w}^{\phi_2}, \hat{w},
        \Omega = e(g, \hat{g})^{\alpha} \big)
    \end{align*}
and the master key as $MK = (\hat{g}, \hat{g}^{\alpha}, \hat{h}, \{ \hat{u}_i
\}_{i=1}^l)$. Additionally, it sets $f = k, \hat{f} = \hat{k}$ for the
semi-functional ciphertext and private key.
Let $\Delta(ID) = y_h + \sum_{i=1}^m y_{u_i} I_i$ where $ID = (I_1, \ldots,
I_m)$. $\mcA$ adaptively requests a private key for $ID = (I_1, \ldots,
I_m)$. If this is a $j$-th private key query, then $\mcB_3$ handles this
query as follows:
\begin{itemize}
\item \textbf{Case $j < k$} : It creates a semi-functional private key by
    calling \textbf{KeyGenSF-2} since it knows the master key and the
    tuple $(\hat{f}^{-\nu}, \hat{f}, 1)$ for the semi-functional private
    key.

\item \textbf{Case $j = k$} : It first selects random exponents $r_1,
    c_1, c_2, \{ c_{3,i} \}_{i=m+1}^l, s_{k,1} \in \Z_p$. It implicitly
    sets $z_{k,1} = x_1,~ \{ z_{k,2,i} = x_{2,i} \}_{i=m+1}^l$ and
    creates the decryption and delegation components of a private key as
    \begin{align*}
    &   K_{1,1} = \hat{g}^{\alpha} (\hat{h} \prod_{i=1}^m \hat{u}_i^{I_i})^{r_1}
            (\hat{w}^{\phi_1})^{c_1} (\hat{k}^{x_1})^{-\nu s_{k,1}},~
        K_{1,2} = (\hat{w}^{\phi_2})^{c_1} (\hat{k}^{x_1})^{s_{k,1}},~
        K_{1,3} = \hat{w}^{c_1},~ \\
    &   K_{2,1} = \hat{g}^{r_1} (\hat{w}^{\phi_1})^{c_2} \hat{k}^{-\nu s_{k,1}},~
        K_{2,2} = (\hat{w}^{\phi_2})^{c_2} \hat{k}^{s_{k,1}},~
        K_{2,3} = \hat{w}^{c_2},~ \\
    &   \big\{
        L_{3,i,1} = \hat{u}_i^{r_1} (\hat{w}^{\phi_1})^{c_{3,i}}
                    (\hat{k}^{x_{2,i}})^{-\nu s_{k,1}},~
        L_{3,i,2} = (\hat{w}^{\phi_2})^{c_{3,i}} (\hat{k}^{x_{2,i}})^{s_{k,1}},~
        L_{3,i,3} = \hat{w}^{c_{3,i}}
        \big\}_{i=m+1}^l.
    \end{align*}
    Next, it selects random exponents $r_2, c_4, c_5, \{ c_{6,i}
    \}_{i=m+1}^l \in \Z_p$. It implicitly sets $s_{k,2} = y$ and creates
    the randomization components of a private key as
    \begin{align*}
    &   R_{1,1} = (\hat{h} \prod_{i=1}^m \hat{u}_i^{I_i})^{r_2}
                  (\hat{w}^{\phi_1})^{c_4} (D_1)^{-\nu},~
        R_{1,2} = (\hat{w}^{\phi_2})^{c_4} D_1,~
        R_{1,3} = \hat{w}^{c_4},~ \\
    &   R_{2,1} = \hat{g}^{r_2} (\hat{w}^{\phi_1})^{c_5} (\hat{k}^y)^{-\nu},~
        R_{2,2} = (\hat{w}^{\phi_2})^{c_5} \hat{k}^y,~
        R_{2,3} = \hat{w}^{c_5},~ \\
    &   \big\{
        R_{3,i,1} = \hat{u}_i^{r_2} (\hat{w}^{\phi_1})^{c_{6,i}} (D_{2,i})^{-\nu},~
        R_{3,i,2} = (\hat{w}^{\phi_2})^{c_{6,i}} D_{2,i},~
        R_{3,i,3} = \hat{w}^{c_{6,i}}
        \big\}_{i=m+1}^l.
    \end{align*}

\item \textbf{Case $j > k$} : It creates a normal private key by calling
    \textbf{KeyGen} since it knows the master key.
\end{itemize}
In the challenge step, $\mcA$ submits two challenge hierarchical identities
$ID_0^* = (I_{0,1}^*, \ldots, I_{0,n}^*), ID_1^* = (I_{1,1}^*, \ldots,
I_{1,n}^*)$ and two messages $M_0^*, M_1^*$. $\mcB_3$ flips a random coin
$\gamma \in \{0,1\}$ internally. It creates a semi-functional challenge
ciphertext by calling $\textbf{EncryptSF}$ on the message $M_{\gamma}$ and
the hierarchical identity $ID_{\gamma}^*$ since it knows the tuple $(1, f,
f^{-\phi_2})$ for the semi-functional ciphertext.
Finally, $\mcA$ outputs a guess $\gamma'$. If $\gamma = \gamma'$, $\mcB_3$
outputs 0. Otherwise, it outputs 1.

\vs \textbf{Analysis}. We first show that the distribution of the simulation
using $D, T_0 = (D_1^0, \ldots, D_{2,l}^0)$ is the same as
$\textbf{G}'_{1,k}$. It is easy to check that the private key components are
correctly distributed except the randomization components of the $k$-th
private key. If we implicitly set
    $z_{k,1} = x_1, \{ z_{k,2,i} = x_{2,i} \}_{i=m+1}^l, s_{k,2} = y$,
then the randomization components of the $k$-th private key have the same
distribution as $\textbf{G}'_{1,k}$.
We next show that the distribution of the simulation using $D, T_1 = (D_1^1,
\ldots, D_{2,l}^1)$ is the same as $\textbf{G}_{1,k}$. We only consider the
distribution of the randomization components of the $k$-th private key since
$T$ is only used in the randomization components of the $k$-th private key.
If we implicitly set
    $s_{k,2} = y,~ z_{k,3} = d_1 / y,~ \{ z_{k,4,i} = d_{2,i} / y \}_{i=m+1}^l$,
then the randomization components are correctly distributed as
    \begin{align*}
    R_{1,1} &= (\hat{h} \prod_{i=1}^m \hat{u}_i^{I_i})^{r_2}
               (\hat{w}^{\phi_1})^{c_4} (\hat{f}^{-\nu})^{s_{k,2} z_{k,3}}
             = (\hat{h} \prod_{i=1}^m \hat{u}_i^{I_i})^{r_2}
               (\hat{w}^{\phi_1})^{c_4} (\hat{k}^{-\nu})^{y \cdot d_1 / y}
             = (\hat{h} \prod_{i=1}^m \hat{u}_i^{I_i})^{r_2}
               (\hat{w}^{\phi_1})^{c_4} (D_1^1)^{-\nu}, \\
    R_{3,i,1} &= \hat{u}_i^{r_2} (\hat{w}^{\phi_1})^{c_{6,i}}
               (\hat{f}^{-\nu})^{s_{k,2} z_{k,4,i}}
             = \hat{u}_i^{r_2} (\hat{w}^{\phi_1})^{c_{6,i}}
               (\hat{k}^{-\nu})^{y \cdot d_{2,i} / y}
             = \hat{u}_i^{r_2} (\hat{w}^{\phi_1})^{c_{6,i}} (D_{2,i}^1)^{-\nu}.
    \end{align*}
From the above analysis, we can obtain $\Pr [\mcB_3(D,T_0) = 0] - 1/2 =
\Adv_{\mcA}^{G'_{1,k}}$ and $\Pr [\mcB_3(D,T_1) = 0] - 1/2 =
\Adv_{\mcA}^{G_{1,k}}$. Thus, we can easily derive the advantage of $\mcB_3$
as
    \begin{align*}
    \Adv_{\mcB_3}^{A3}(\lambda)
     =  \big| \Pr[\mcB_3(D, T_0) = 0] - \Pr[\mcB_3(D, T_1) = 0] \big|
     =  \big| \Adv_{\mcA}^{G'_{1,k}} - \Adv_{\mcA}^{G_{1,k}} \big|.
    \end{align*}
This completes our proof.

\subsection{Proof of Lemma~\ref{lem:ahibe-prime-4}
(Indistinguishability of $\textbf{G}_2$ and $\textbf{G}_3$)}


In this proof, private keys and the challenge ciphertext are semi-functional
type-2 and semi-functional respectively, but a session key should be correct
or random depending on the $T$ value of the given assumption. The main idea
of this proof is to enforce a simulator to solve the Computational
Diffie-Hellman (CDH) problem in order to create the normal types of private
keys and ciphertexts. However, the simulator can generate the semi-functional
types of private keys and ciphertexts since an additional random value in
semi-functional types enables the cancellation of the CDH value.

\vs \textbf{Simulator}. Suppose there exists an adversary $\mcA$ that
distinguishes between $\textbf{G}_2$ and $\textbf{G}_3$ with a non-negligible
advantage. A simulator $\mcB_4$ that breaks Assumption 4 using $\mcA$ is
given: a challenge tuple
    $D = ((p, \G, \hat{\G}, \G_T, e),
    k, k^a, k^b, k^c, \hat{k}, \hat{k}^a, \hat{k}^b, \hat{k}^c)$ and $T$
    where $T = T_0 = e(k, \hat{k})^{abc}$ or $T = T_1 = e(k, \hat{k})^d$.
Then $\mcB_4$ that interacts with $\mcA$ is described as follows:
$\mcB_4$ first chooses random exponents $\phi_1, \phi_2 \in \Z_p$ and random
blinding values $y_g, y_h, \{ y_{u_i} \}_{i=1}^l, y_w \in \Z_p$. It sets
    $g = k^{y_g}, h = k^{y_h}, \{ u_i = k^{y_{u_i}} \}_{i=1}^l,
    \hat{g} = \hat{k}^{y_g}, \hat{h} = \hat{k}^{y_h},
    \{ \hat{u}_i = \hat{k}^{y_{u_i}} \}_{i=1}^l, \hat{w} = \hat{k}^{y_w}$.
It implicitly sets $\nu = a, \tau = \phi_1 + a \phi_2, \alpha = ab$ and
creates the public parameters as
    \begin{align*}
    &   g,~ g^{\nu} = (k^a)^{y_g},~ g^{-\tau} = k^{-y_g \phi_1} (k^a)^{-y_g \phi_2},~
        h,~ h^{\nu} = (k^a)^{y_h},~ h^{-\tau} = k^{-y_h \phi_1} (k^a)^{-y_h \phi_2},~ \\
    &   \big\{
        u_i,~ u_i^{\nu} = (k^a)^{y_{u_i}},~
        u_i^{-\tau} = k^{-y_{u_i} \phi_1} (k^a)^{-y_{u_i} \phi_2}
        \big\}_{i=1}^l,~
        \hat{w}^{\phi_1}, \hat{w}^{\phi_2}, \hat{w},~
        \Omega = e(k^a, \hat{k}^b)^{y_g^2}.
    \end{align*}
Additionally, it sets $f = k, \hat{f} = \hat{k}$ for the semi-functional
ciphertext and private key.
Let $\Delta(ID) = y_h + \sum_{i=1}^m y_{u_i} I_i$ where $ID = (I_1, \ldots,
I_m)$. $\mcA$ adaptively requests a private key for $ID = (I_1, \ldots,
I_m)$. To response the private key query, $\mcB_4$ first selects random
exponents $r_1, c_1, c_2, \{ c_{3,i} \}_{i=m+1}^l, s_{k,1}, z'_{k,1}, \{
z_{k,2,i} \}_{i=m+1}^l \in \Z_p$. It implicitly sets
    $z_{k,1} = b y_g / s_{k,1} + z'_{k,1}$
and creates the decryption and delegation components of a semi-functional
private key as
    \begin{align*}
    &   K_{1,1} = (\hat{h} \prod_{i=1}^m \hat{u}_i^{I_i})^{r_1}
                  (\hat{w}^{\phi_1})^{c_1} (\hat{k}^a)^{-s_{k,1} z'_{k,1}},~
        K_{1,2} = (\hat{w}^{\phi_2})^{c_1} (\hat{k}^b)^{y_g}
                  \hat{k}^{s_{k,1} z'_{k,1}},~
        K_{1,3} = \hat{w}^{c_1}, \\
    &   K_{2,1} = \hat{g}^{r_1} (\hat{w}^{\phi_1})^{c_2} (\hat{k}^a)^{-s_{k,1}},~
        K_{2,2} = (\hat{w}^{\phi_2})^{c_2} \hat{k}^{s_{k,1}},~
        K_{2,3} = \hat{w}^{c_2},~ \\
    &   \big\{
        L_{3,i,1} = \hat{u}_i^{r_1} (\hat{w}^{\phi_1})^{c_{3,i}}
                    (\hat{k}^a)^{-s_{k,1} z_{k,2,i}},~
        L_{3,i,2} = (\hat{w}^{\phi_2})^{c_{3,i}} \hat{k}^{s_{k,1} z_{k,2,i}},~
        L_{3,i,3} = \hat{w}^{c_{3,i}}
        \big\}_{i=m+1}^l.
    \end{align*}
Next, it selects random exponents $r_2, c_4, c_5, \{ c_{6,i} \}_{i=m+1}^l,
s_{k,2}, z_{k,3}, \{ z_{k,4,i} \}_{i=m+1}^l \in \Z_p$ and creates the
randomization components of a semi-functional private key.
In the challenge step, $\mcA$ submits two challenge hierarchical identities
$ID_0^* = (I_{0,1}^*, \ldots, I_{0,n}^*), ID_1^* = (I_{1,1}^*, \ldots,
I_{1,n}^*)$ and two messages $M_0^*, M_1^*$. $\mcB_4$ flips a random coin
$\gamma \in \{0,1\}$ internally and chooses random exponents $s'_c, z'_c \in
\Z_p$. It implicitly sets
    $t = c,~
    s_c = -ac y_g + s'_c,~
    z_c = -ac \Delta(ID_{\gamma}^*) / s_c + z'_c / s_c$
and creates the semi-functional ciphertext as
    \begin{align*}
    &   C       = (T)^{y_g^2} \cdot M_{\gamma}^*,~
        C_{1,1} = (k^c)^{y_g},~
        C_{1,2} = k^{s'_c},~
        C_{1,3} = (k^c)^{-y_g \phi_1} k^{-\phi_2 s'_c}, \\
    &   C_{2,1} = (k^c)^{\Delta(ID_{\gamma}^*)},~
        C_{2,2} = k^{z'_c},~
        C_{2,3} = (k^c)^{-\Delta(ID_{\gamma}^*) \phi_1} k^{-\phi_2 z'_c}.
    \end{align*}
Finally, $\mcA$ outputs a guess $\gamma'$. If $\gamma = \gamma'$, $\mcB_4$
outputs 0. Otherwise, it outputs 1.

\vs \textbf{Analysis}. We first show that the distribution of the simulation
using $D, T_0 = e(k,\hat{k})^{abc}$ is the same as $\textbf{G}_2$. The public
parameters are correctly distributed since the random blinding values $y_g,
y_h, \{ y_{u_i} \}, y_w$ are used. The semi-functional private key is
correctly distributed as
    \begin{align*}
    K_{1,1} &= \hat{g}^{\alpha} (\hat{h} \prod_{i=1}^m \hat{u}_i^{I_i})^{r_1}
               (\hat{w}^{\phi_1})^{c_1} (\hat{f}^{-\nu})^{s_{k,1} z_{k,1}}
             = \hat{k}^{y_g ab} (\hat{h} \prod_{i=1}^m \hat{u}_i^{I_i})^{r_1}
               (\hat{w}^{\phi_1})^{c_1}
               (\hat{k}^{-a})^{s_{k,1} \cdot (b y_g / s_{k,1} + z'_{k,1})} \\
            &= (\hat{h} \prod_{i=1}^m \hat{u}_i^{I_i})^{r_1} (\hat{w}^{\phi_1})^{c_1}
               (\hat{k}^a)^{-s_{k,1} z'_{k,1}}.
    \end{align*}
Note that it can only create a semi-functional private key since $z_{k,1} = b
y_g / s_{k,1} + z'_{k,1}$ enables the cancellation of $\hat{k}^{ab}$. The
semi-functional challenge ciphertext is correctly distributed as
    \begin{align*}
    C ~~~~  &= e(g, \hat{g})^{\alpha t} M_{\gamma}^*
             = e(k^{y_g}, \hat{k}^{y_g})^{ab c} M_{\gamma}^*
             = (T)^{y_g^2} M_{\gamma}^*, \\
    C_{1,1} &= g^t = (k^{y_g})^c = (k^c)^{y_g}, \\
    C_{1,2} &= (g^{\nu})^t f^{s_c}
             = (k^{y_g a})^c k^{-ac y_g + s'_c}
             = k^{s'_c}, \\
    C_{1,3} &= (g^{-\tau})^t (f^{-\phi_2})^{s_c}
             = (k^{-y_g (\phi_1 + a \phi_2)})^{c} k^{-\phi_2 (-ac y_g + s'_c)}
             = (k^c)^{-y_g \phi_1} k^{-\phi_2 s'_c},
             \displaybreak[0] \\
    C_{2,1} &= (h \prod_{i=1}^n u_i^{I_{\gamma,i}^*})^t
             = (k^{y_h} \prod_{i=1}^n k^{y_{u_i} I_{\gamma,i}^*})^c
             = (k^c)^{\Delta(ID_{\gamma}^*)}, \\
    C_{2,2} &= (h^{\nu} \prod_{i=1}^n (u_i^{\nu})^{I_{\gamma,i}^*})^t f^{s_c z_c}
             = (k^{y_h a} \prod_{i=1}^n k^{y_{u_i} a I_{\gamma,i}^*})^c
               k^{s_c (-ac \Delta(ID_{\gamma}^*) / s_c + z'_c / s_c)}
             = k^{z'_c},
             \displaybreak[0] \\
    C_{2,3} &= (h^{-\tau} \prod_{i=1}^n (u_i^{-\tau})^{I_{\gamma,i}^*})^t
               (f^{-\phi_2})^{s_c z_c} \\
            &= (k^{-y_h (\phi_1 + a \phi_2)} \prod_{i=1}^n
                k^{-y_{u_i} (\phi_1 + a \phi_2) I_{\gamma,i}^*})^c
                (k^{-\phi_2})^{s_c (-ac \Delta(ID_{\gamma}^*) / s_c + z'_c / s_c)}
             = (k^c)^{-\Delta(ID_{\gamma}^*) \phi_1} k^{-\phi_2 z'_c}.
    \end{align*}
Note that it can create a semi-functional ciphertext since $s_c, z_c$ enable
the cancellation of $k^{ac}$.
We next show that the distribution of the simulation using $D, T_1 =
e(k,\hat{k})^d$ is the same as $\textbf{G}_3$. It is obvious that $C$ is a
random element since $T_1 = e(k,\hat{k})^d$. From the above analysis, we
obtain $\Pr [\mcB_4 (D,T_0) = 0] - 1/2 = \Adv_{\mcA}^{G_2}$ and $\Pr [\mcB_4
(D,T_1) = 0] - 1/2 = \Adv_{\mcA}^{G_3}$. Thus, we can easily derive the
advantage of $\mcB_4$ as
    \begin{align*}
    \Adv_{\mcB_4}^{A4}(\lambda)
     =  \big| \Pr[\mcB_4 (D, T_0) = 0] - \Pr[\mcB_4 (D, T_1) = 0] \big|
     =  \big| \Adv_{\mcA}^{G_2} - \Adv_{\mcA}^{G_3} \big|.
    \end{align*}
This completes our proof.

\subsection{Proof of Lemma~\ref{lem:ahibe-prime-5}
(Indistinguishability of $\textbf{G}_3$ and $\textbf{G}_4$)}


In this proof, private keys and the challenge ciphertext are semi-functional
type-2 and semi-functional respectively, and the elements of the challenge
ciphertext should be well-formed or random depending on the $T$ value of the
given assumption. The idea to generate semi-functional type-2 private keys
and semi-functional ciphertexts is similar to Lemma~\ref{lem:ahibe-prime-4},
but it uses a different assumption. To prove anonymity, the simulator embeds
the $T$ value of the assumption into the all elements of the challenge
ciphertext that contains an identity.

\vs \textbf{Simulator}. Suppose there exists an adversary $\mcA$ that
distinguishes between $\textbf{G}_3$ and $\textbf{G}_4$ with a non-negligible
advantage. A simulator $\mcB_5$ that breaks Assumption 5 using $\mcA$ is
given: a challenge tuple
    $D = ((p, \G, \hat{\G}, \G_T, e),
    k, k^a, k^b, k^c, k^{ab}, k^{a^2 b}, \hat{k}, \hat{k}^a, \hat{k}^b)$ and
    $T$ where $T = T_0 = k^{abc}$ or $T = T_1 = k^d$.
Then $\mcB_5$ that interacts with $\mcA$ is described as follows:
$\mcB_5$ first chooses random exponents $\phi_1, \phi_2, \alpha \in \Z_p$ and
random blinding values $y_g, y_h, \{ y_{u_i} \}_{i=1}^l, y_w \in \Z_p$. It
sets
    $g = k^{y_g}, h = (k^{ab})^{y_h}, \{ u_i = (k^{ab})^{y_{u_i}} \}_{i=1}^l,
    \hat{g} = \hat{k}^{y_g}, \hat{w} = \hat{k}^{y_w},
    \hat{g}^{\alpha} = \hat{k}^{y_g \alpha}$.
It implicitly sets $\nu = a, \tau = \phi_1 + a \phi_2$ and
publishes the public parameters as
    \begin{align*}
    &   g,~ g^{\nu} = (k^a)^{y_g},~
        g^{-\tau} = k^{-y_g \phi_1} (k^a)^{-y_g \phi_2},~
        h,~ h^{\nu} = (k^{a^2b})^{y_h},~
        h^{-\tau} = (k^{ab})^{-y_h \phi_1} (k^{a^2b})^{-y_h \phi_2},~ \\
    &   \big\{
        u_i, u_i^{\nu} = (k^{a^2 b})^{y_{u_i}},
        u_i^{-\tau} = (k^{ab})^{-y_{u_i} \phi_1} (k^{a^2b})^{-y_{u_i} \phi_2}
        \big\}_{i=1}^l,~
        \hat{w}^{\phi_1}, \hat{w}^{\phi_2}, \hat{w},~
        \Omega = e(k, \hat{k})^{y_g^2 \alpha}.
    \end{align*}
It also implicitly sets $\hat{h} = (\hat{k}^{ab})^{y_h}, \{ \hat{u}_i =
(\hat{k}^{ab})^{y_{u_i}} \}$ for the master key, but it cannot create these
values since $\hat{k}^{ab}$ is not given. Additionally, it sets $f = k,
\hat{f} = \hat{k}$ for the semi-functional ciphertext and private key.
Let $\Delta(ID) = y_h + \sum_{i=1}^m y_{u_i} I_i$ where $ID = (I_1, \ldots,
I_m)$. $\mcA$ adaptively requests a private key for $ID = (I_1, \ldots,
I_m)$. To response the private key query, $\mcB_5$ first selects random
exponents $r_1, c_1, c_2, \{ c_{3,i} \}_{i=m+1}^l, s_{k,1}, z'_{k,1},
\linebreak[1] \{ z'_{k,2,i} \}_{i=m+1}^l \in \Z_p$. It implicitly sets
    $z_{k,1} = b \Delta(ID) r_1 / s_{k,1} + z'_{k,1},~
    \{ z_{k,2,i} = b y_{u_i} r_1 / s_{k,1} + z'_{k,2,i} \}_{i=m+1}^l$
and creates the decryption and delegation components of a semi-functional
private key as
    \begin{align*}
    &   K_{1,1} = \hat{g}^{\alpha} (\hat{w}^{\phi_1})^{c_1}
                  (\hat{k}^a)^{-s_{k,1} z'_{k,1}},~
        K_{1,2} = (\hat{w}^{\phi_2})^{c_1}
                  (\hat{k}^b)^{\Delta(ID) r_1}
                  \hat{k}^{s_{k,1} z'_{k,1}},~
        K_{1,3} = \hat{w}^{c_1}, \\
    &   K_{2,1} = \hat{g}^{r_1} (\hat{w}^{\phi_1})^{c_2} (\hat{k}^a)^{-s_{k,1}},~
        K_{2,2} = (\hat{w}^{\phi_2})^{c_2} \hat{k}^{s_{k,1}},~
        K_{2,3} = \hat{w}^{c_2},~ \\
    &   \big\{
        L_{3,i,1} = (\hat{w}^{\phi_1})^{c_{3,i}} (\hat{k}^a)^{-s_{k,1} z'_{k,2,i}},~
        L_{3,i,2} = (\hat{w}^{\phi_2})^{c_{3,i}} (\hat{k}^b)^{y_{u_i} r_1}
                    \hat{k}^{s_{k,1} z'_{k,2,i}},~
        L_{3,i,3} = \hat{w}^{c_{3,i}}
        \big\}_{i=m+1}^l.
    \end{align*}
Next, it selects random exponents $r_2, c_4, c_5, \{ c_{6,i} \}_{i=m+1}^l,
s_{k,2}, z'_{k,3}, \{ z'_{k,4,i} \}_{i=m+1}^l \in \Z_p$ and creates the
randomization components of a semi-functional private key by implicitly
setting
    $z_{k,3} = b \Delta(ID) r_2 / s_{k,2} + z'_{k,3},~
    \{ z_{k,4,i} = b y_{u_i} r_2 / s_{k,2} + z'_{k,4,i} \}_{i=m+1}^l$.
We omit the detailed description of these, since these are similar to the
decryption and delegation components except that $R_{1,1}$ does not have
$\hat{g}^{\alpha}$.
In the challenge step, $\mcA$ submits two challenge hierarchical identities
$ID_0^* = (I_{0,1}^*, \ldots, I_{0,n}^*), ID_1^* = (I_{1,1}^*, \ldots,
I_{1,n}^*)$ and two messages $M_0^*, M_1^*$. $\mcB_5$ flips a random coin
$\gamma \in \{0,1\}$ internally and chooses random exponents $\delta, s'_c,
z'_c \in \Z_p$. It implicitly sets
    $t = c,~ s_c = -ac y_g + s'_c,~
    z_c = -a^2bc \Delta(ID_{\gamma}^*) / s_c + abc z'_c / s_c$
and creates the semi-functional ciphertext as
    \begin{align*}
    &   C       = \Omega^{\delta} \cdot M_{\gamma}^*,~
        C_{1,1} = (k^c)^{y_g},~
        C_{1,2} = (k^a)^{s'_c},~
        C_{1,3} = (k^c)^{-y_g \phi_1} k^{-\phi_2 s'_c}, \\
    &   C_{2,1} = (T)^{\Delta(ID_{\gamma}^*)},~
        C_{2,2} = (T)^{z'_c},~
        C_{2,3} = (T)^{-\Delta(ID_{\gamma}^*) \phi_1}
                  (T)^{-z'_c \phi_2}.
    \end{align*}
Finally, $\mcA$ outputs a guess $\gamma'$. If $\gamma = \gamma'$, $\mcB_5$
outputs 0. Otherwise, it outputs 1.

\vs \textbf{Analysis}. We first show that the distribution of the simulation
using $D, T_0 = k^{abc}$ is the same as $\textbf{G}_3$. The public parameters
are correctly distributed since the random blinding values are used. The
semi-functional private key is correctly distributed as
    \begin{align*}
    K_{1,1} &= \hat{g}^{\alpha} (\hat{h} \prod_{i=1}^m \hat{u}_i^{I_i})^{r_1}
               (\hat{w}^{\phi_1})^{c_1} (\hat{f}^{-\nu})^{s_{k,1} z_{k,1}}
             = \hat{g}^{\alpha} (\hat{k}^{ab})^{\Delta(ID) r_1}
               (\hat{w}^{\phi_1})^{c_1}
               (\hat{k}^{-a})^{s_{k,1} (b \Delta(ID) r_1 / s_{k,1} + z'_{k,1})} \\
            &= \hat{g}^{\alpha} (\hat{w}^{\phi_1})^{c_1} (\hat{k}^a)^{-s_{k,1} z'_{k,1}},
            \displaybreak[0] \\
    K_{2,1} &= \hat{g}^{r_1} (\hat{w}^{\phi_1})^{c_2} (\hat{f}^{-\nu})^{s_{k,1}}
             = \hat{g}^{r_1} (\hat{w}^{\phi_1})^{c_2} (\hat{k}^a)^{-s_{k,1}}, \\
    L_{3,i,1} &= \hat{u}_i^{r_1} (\hat{w}^{\phi_1})^{c_{3,i}}
               (\hat{f}^{-\nu})^{s_{k,1} z_{k,2,i}}
             = (\hat{k}^{ab})^{y_{u_i} r_1} (\hat{w}^{\phi_1})^{c_{3,i}}
               (\hat{k}^{a})^{-s_{k,1} \cdot (b y_{u_i} r_1 / s_{k,1} + z'_{k,2,i})} \\
            &= (\hat{w}^{\phi_1})^{c_{3,i}} (\hat{k}^a)^{-s_{k,1} z'_{k,2,i}}.
    \end{align*}
Note that it can only create a semi-functional type-2 private key since
$z_{k,1}, \{z_{k,2,i} \}, z_{k,3}, \{ z_{k,4,i} \}$ enable the cancellation
of $\hat{k}^{ab}$. The semi-functional challenge ciphertext is correctly
distributed as
    \begin{align*}
    C_{1,1} &= g^t = (k^{y_g})^c = (k^c)^{y_g}, \\
    C_{1,2} &= (g^{\nu})^t f^{s_c}
             = (k^{y_g a})^c k^{-ac y_g + s'_c}
             = k^{s'_c}, \\
    C_{1,3} &= (g^{-\tau})^t (f^{-\phi_2})^{s_c}
             = (k^{-y_g (\phi_1 + a \phi_2)})^{c} k^{-\phi_2 (-ac y_g + s'_c)}
             = (k^c)^{-y_g \phi_1} k^{-\phi_2 s'_c},
             \displaybreak[0] \\
    C_{2,1} &= (h \prod_{i=1}^n u_i^{I_{\gamma,i}^*})^t
             = (k^{ab})^{\Delta(ID_{\gamma}^*) c}
             = (T_0)^{\Delta(ID_{\gamma}^*)},~ \\
    C_{2,2} &= (h \prod_{i=1}^n u_i^{I_{\gamma,i}^*})^{\nu t} f^{s_c z_c}
             = ((k^{ab})^{\Delta(ID_{\gamma}^*)})^{ac}
               k^{s_c (-a^2bc \Delta(ID_{\gamma}^*) / s_c + abc z'_c / s_c)}
             = (T_0)^{z'_c},~ \\
    C_{2,3} &= (h \prod_{i=1}^n u_i^{I_{\gamma,i}^*})^{-\tau t} (f^{-\phi_2})^{s_c z_c}
             = ((k^{ab})^{\Delta(ID_{\gamma}^*)})^{-(\phi_1 + a \phi_2)c}
               k^{-\phi_2 s_c (-a^2bc \Delta(ID_{\gamma}^*) / s_c + abc z'_c / s_c)} \\
            &= (T_0)^{- \Delta(ID_{\gamma}^*) \phi_1}
               (T_0)^{- z'_c \phi_2}.
    \end{align*}
Note that it can only create a semi-functional ciphertext since $s_c, z_c$
enable the cancellation of $k^{a^2bc}$.
We next show that the distribution of the simulation using $D, T_1 = k^d$ is
the same as $\textbf{G}_4$. We only consider $C_{2,1}, C_{2,2}, C_{2,3}$
components of the semi-functional challenge ciphertext since $T$ is used for
these components. If we implicitly sets $P = k^{\Delta(ID_{\gamma}^*) d/c}$
and $z_c = -ad \Delta(ID_{\gamma}^*) / s_c + d z'_c / s_c$, then the
semi-functional challenge ciphertext is correctly distributed as
    \begin{align*}
    C_{2,1} &= P^c
             = (k^{\Delta(ID_{\gamma}^*) d/c})^c
             = (T_1)^{\Delta(ID_{\gamma}^*)},~ \\
    C_{2,2} &= P^{\nu c} f^{s_c z_c}
             = (k^{\Delta(ID_{\gamma}^*) d/c})^{ac}
               k^{s_c (-ad \Delta(ID_{\gamma}^*) / s_c + d z'_c / s_c)}
             = (T_1)^{z'_c},~ \\
    C_{2,3} &= P^{-\tau c} (f^{-\phi_2})^{s_c z_c}
             = (k^{\Delta(ID_{\gamma}^*) d/c})^{-(\phi_1 + a \phi_2)c}
               k^{-\phi_2 s_c (-ad \Delta(ID_{\gamma}^*) / s_c + d z'_c / s_c)} \\
            &= (T_1)^{- \Delta(ID_{\gamma}^*) \phi_1}
               (T_1)^{- z'_c \phi_2}.
    \end{align*}
From the above analysis, we obtain $\Pr [\mcB_5(D,T_0) = 0] - 1/2 =
\Adv_{\mcA}^{G_3}$ and $\Pr [\mcB_5 (D,T_1) = 0] - 1/2 = \Adv_{\mcA}^{G_4}$.
Thus, we can easily derive the advantage of $\mcB_5$ as
    \begin{align*}
    \Adv_{\mcB_5}^{A5}(\lambda)
    &=  \big| \Pr[\mcB_5 (D, T_0) = 0] - \Pr[\mcB_5 (D, T_1) = 0] \big|
     =  \big| \Adv_{\mcA}^{G_3} - \Adv_{\mcA}^{G_4} \big|.
    \end{align*}
This completes our proof.

\section{Generic Group Model} \label{sec:gen-mod}

In this section, we prove that the new assumption of this paper is secure
under the generic group model. The generic group model was introduced by
Shoup \cite{Shoup97}, and it is a tool for analyzing generic algorithms that
work independently of the group representation. In the generic group model,
an adversary is given a random encoding of a group element or an arbitrary
index of a group element instead of the actual representation of a group
element. Thus, the adversary performs group operations through oracles that
are provided by a simulator, and the adversary only can check the equality of
group elements. The detailed explanation of the generic group model is given
in \cite{BonehBG05,KatzSW08}.

\subsection{Master Theorem}

To analyze the new assumption of this paper, we slightly modify the master
theorem of Katz et al. \cite{KatzSW08} since the new assumption is defined
over asymmetric bilinear groups of prime order.
Let $\G, \hat{\G}, \G_T$ be asymmetric bilinear groups of prime order $p$.
The bilinear map is defined as $e:\G \times \hat{\G} \rightarrow \G_T$. In
the generic group model, a random group element of $\G, \hat{\G}, \G_T$ is
represented as a random variable $P_i, Q_i, R_i$ respectively where $P_i,
Q_i, R_i$ are chosen uniformly in $\Z_p$. We say that a random variable has
degree $t$ if the maximum degree of any variable is $t$. The generalized
definition of dependence and independence is given as follows:

\begin{definition} \label{def-dep-indep}
Let $P = \{ P_1, \ldots, P_u \},~ T_0, T_1$ be random variables over $\G$
where $T_0 \neq T_1$, let $Q = \{ Q_1, \ldots, Q_w \}$ be random variables
over $\hat{\G}$, and let $R = \{ R_1, \ldots, R_v \}$ be random variables
over $\G_T$. Let $l = \max \{ u, w, v \}$. We say that $T_b$ is dependent on
$P$ if there exists constants $\alpha, \{\beta_i\}$ such that
    \begin{align*}
    \alpha \cdot T_b = \sum_{i=1}^u \beta_i \cdot P_i
    \end{align*}
where $\alpha \neq 0$. We say that $T_b$ is independent of $P$ if $T_b$ is
not dependent on $P$.
We say that $\{ e(T_b, Q_i) \}_{i}$ is dependent on $P \cup Q \cup R$ if
there exist constants $\{\alpha_{i}\}, \{\beta_{i,j}\}, \{\gamma_{i}\}$ such
that
    \begin{align*}
    \sum_{i=1}^w \alpha_{i} \cdot e(T_b, Q_i)
    =  \sum_{i=1}^u \sum_{j=1}^w \beta_{i,j} \cdot e(P_i, Q_j) +
       \sum_{i=1}^v \gamma_i \cdot R_i
    \end{align*}
where $\alpha_{i} \neq 0$ for at least one $i$. We say that $\{ e(T_b, Q_i)
\}_{i}$ is independent of $P \cup Q \cup R$ if $\{ e(T_b, Q_i) \}_{i}$ is not
dependent on $P \cup Q \cup R$.
\end{definition}

We can obtain the following theorem by using the above dependence and
independence of random variables.

\begin{theorem} \label{thm-master}
Let $P = \{ P_1, \ldots, P_u \},~ T_0, T_1$ be random variables over $\G$
where $T_0 \neq T_1$, let $Q = \{ Q_1, \ldots, Q_w \}$ be random variables
over $\hat{\G}$, and let $R = \{R_1, \ldots, R_v\}$ be random variables over
$\G_T$. Let $l = \max \{ u, w, v \}$. Consider the following experiment in
the generic group model:
\begin{quote}
An algorithm is given $P = \{ P_1, \ldots, P_u \}$, $Q = \{ Q_1, \ldots, Q_w
\}$, and $R = \{ R_1, \ldots, R_v \}$. A random bit $b$ is chosen, and the
adversary is given $T_b$. The algorithm outputs a bit $b'$, and succeeds if
$b'=b$. The algorithm's advantage is the absolute value of the difference
between its success probability and $1/2$.
\end{quote}
If $T_b$ is independent of $P$ for all $b \in \{0,1\}$, and $\{ e(T_b, Q_j)
\}_{j}$ is independent of $P \cup Q \cup R$ for all $b \in \{0,1\}$, then any
algorithm $\mcA$ issuing at most $q$ instructions has an advantage at most
$3(q+2l)^2 t/p$.
\end{theorem}

\begin{proof}
The proof consists of a sequence of games. The first game will be the
original experiment that is described in the theorem and the last game will
be a game that the algorithm has no advantage. We define the games as
follows:
\begin{description}
\item [\textbf{Game} $\textbf{G}_1$.] This game is the original game. In
    this game, the simulator instantiates each of random variables $P, Q,
    R, T_b$ by choosing random values for each of the formal variables.
    Then it gives the handles of $P, Q, R, T_b$ to the algorithm $\mcA$.
    Next, $\mcA$ requests a sequence of multiplication, exponentiation,
    and pairing instructions, and is given the handles of results.
    Finally, $\mcA$ outputs a bit $b'$.
\item [\textbf{Game} $\textbf{G}_2$.] We slightly modify $\textbf{G}_1$
    into a new game $\textbf{G}_2$. In this game, the simulator never
    concretely instantiates the formal variables. Instead it keeps the
    formal polynomials themselves. Additionally, the simulator gives
    identical handles for two elements only if these elements are equal
    as formal polynomials in each of their components. That is, the
    simulator of this game assigns different handles for $X $ and $Y $
    since these are different polynomials. Note that the simulator of
    $\textbf{G}_1$ assigned the same handle for $X=(X_1, \ldots, X_n)$
    and $Y=(Y_1, \ldots, Y_n)$ if $X_i = Y_i$ for all $i$.
\end{description}

To prove the theorem, we will show that the statistical distance between two
games $\textbf{G}_1$ and $\textbf{G}_2$ is negligible and the advantage of
the algorithm in $\textbf{G}_2$ is zero. Then the advantage of the algorithm
in the original game is bounded by the statistical distance between two
games.

We first show that the statistical distance between two games $\textbf{G}_1$
and $\textbf{G}_2$ is negligible. The only difference between two games is
the case that two different formal polynomials take the same value by
concrete instantiation. The probability of this event is at most $t/p$ from
the Schwartz-Zippel Lemma \cite{Schwartz80}. If we consider all pairs of
elements produced by the algorithm $\mcA$, the statistical distance between
two games is at most $3(q+2l)^2 t/p$ since $\mcA$ can request at most $q$
instructions, the maximum size of handles in each group is at most $q+ 2l$,
and there are three different groups.

We next show that the advantage of the algorithm in $\textbf{G}_2$ is zero.
In this game, the algorithm $\mcA$ only can distinguish whether it is given
$T_0$ or $T_1$ if it can generate a formal polynomial that is symbolically
equivalent to some previously generated polynomial for one value of $b$ but
not the other. In this case, we have
    $\alpha \cdot T_b = \sum_{i=1}^u \beta_i \cdot P_i$
where $\alpha \neq 0$, or else we have
    $\sum_{i=1}^w \alpha_{i} \cdot e(T_b, Q_i)
    = \sum_{i=1}^u \sum_{j=1}^w \beta_{i,j} \cdot e(P_i, Q_j) +
      \sum_{i=1}^v \gamma_i \cdot R_i$
where $\alpha_{i} \neq 0$ for at least one $i$ (otherwise, symbolic equality
would hold for both value of $b$). However, the above equations are
contradict to the independence assumptions of the theorem. Therefore, the
advantage of $\mcA$ in this game is zero.
\end{proof}

\subsection{Analysis of Asymmetric 3-Party Diffie-Hellman}

To apply the master theorem of the previous section, we only need to show the
independence of $T_0, T_1$ random variables. Using the notation of previous
section, Assumption 5 (Asymmetric 3-Party Diffie-Hellman) can be written as
    \begin{align*}
    P   &= \{ 1, A, B, C, AB, A^2B \},~
    Q    = \{ 1, A, B \},~
    R    = \{ 1 \},~
    T_0  = ABC,~
    T_1  = D.
    \end{align*}

At first, we show the independence of $T_1$. It is trivial that $T_1$ is
independent of $P$ since a random variable $D$ does not exist in $P$. It is
easy to show that $\{ e(T_1, Q_i) \}_i$ is independent of $P \cup Q \cup R$
since $T_1$ contains a random variable $D$ that does not exist in $P, Q, R$.
Next, we show the independence of $T_0$. It is easy to show that $T_0$ is
independent of $P$ since the random variables with degree 3 are different. To
show the independence of $\{ e(T_0, Q_i) \}_i$, we can derive the sets of
random variables as
    \begin{align*}
    & \{ e(T_0, Q_j) \}_{j}   = \{ ABC, A^2BC, AB^2C \}, \\
    & \{ e(P_i, Q_j) \}_{i,j} = \{ 1, A, B, C, AB, A^2B, A^2, AC, A^3B, B^2,
                                   BC, AB^2, A^2B^2 \}, \\
    & \{ R_i \}               = \{ 1 \}.
    \end{align*}
The random variables of $\{ e(T_0, Q_i) \}_i$ always contain $C$ and the
degree of these random variables is greater than $3$. However, the random
variables of $\{ e(P_i, Q_j) \}_{i,j}$ that contain $C$ have the degree at
most $2$. Thus $\{ e(T_0, Q_i) \}_i$ is independent of $P \cup Q \cup R$.

\section{Conclusion}

In this paper, we proposed an efficient anonymous HIBE scheme with short
ciphertexts and proved its full model security under static assumptions.
Though our construction is based on the IBE scheme of Lewko and Waters
\cite{LewkoW10}, it was not trivial to construct an anonymous HIBE scheme,
since the randomization components of private keys cause a problem in the
security proof of dual system encryption. We leave it as an interesting
problem to construct a fully secure and anonymous HIBE scheme with short
ciphertexts under standard assumptions.

%

\bibliographystyle{plain}
\bibliography{full-ahibe-prime}

\end{document}